\documentclass[11pt,a4paper]{amsart}

\usepackage[left=1in,right=1in,top=1in,bottom=1in,foot=0.5in]{geometry}

\usepackage{amsmath,amsfonts,amssymb}
\usepackage[T1]{fontenc}
\usepackage{color}
\usepackage{epsfig}
\usepackage{graphicx}
\usepackage{cite,url}
\usepackage{hyperref}
\usepackage{cleveref}
\usepackage{xspace}
\usepackage{enumerate}

\usepackage[algo2e,algoruled,vlined,english]{algorithm2e} % typeset algorithms

\newtheorem{thml}{Theorem}%[section]

\newtheorem{mytheorem}{Theorem}[section]
\newtheorem{mylemma}[mytheorem]{Lemma}
\newtheorem{myproposition}[mytheorem]{Proposition}

\newtheorem{myclaim}[mytheorem]{Claim}

\theoremstyle{definition}

\newcommand{\cO}{\ensuremath{\mathcal{O}}\xspace}

\DeclareMathOperator{\Supp}{Supp}

\DeclareMathOperator{\WP}{WP}

\DeclareMathOperator{\ecc}{ecc}

\newcommand{\TC}{\mathrm{(TC)}}

\DeclareMathOperator{\diam}{diam}
\DeclareMathOperator{\rad}{rad}

\newcommand{\commentout}[1]{}

\newcommand{\defoptproblem}[3]{
 \vspace{1mm}
\noindent\fbox{
 \begin{minipage}{0.96\textwidth}
 #1 \\
 {\bf{Input:}} #2 \\
 {\bf{Output:}} #3
 \end{minipage}
 }
 \vspace{1mm}
}

\begin{document}

\title{Beyond Trees: The Weighted Center Problem on Gromov Hyperbolic Graphs}

\author[G.\ Ducoffe]{Guillaume Ducoffe} \address{
University of Bucharest, Faculty of Mathematics and Computer Science
\and National Institute for Research and Development in Informatics, Bucureşti, Rom\^{a}nia} \email{guillaume.ducoffe@ici.ro}

\date{}

\begin{abstract}
The \textsc{Weighted Center} problem takes as its input a graph $G=(V,E)$ together with a profile $\pi$ such that every vertex $v$ is mapped to some nonnegative multiplicative weight $\pi(v)$. 
Its output must be some vertex $c$ minimizing $\max\{\pi(v)d_G(c,v) : v \in V\}$.
The classic \textsc{Center} problem corresponds to the case where $\pi(v) =1$ for every vertex $v$. 
In the literature, various almost linear-time algorithms have been proposed for the \textsc{Center} problem on some well-structured classes of graphs.
By contrast, similarly efficient algorithms for the \textsc{Weighted Center} problem have been scarce. We investigate how the Gromov hyperbolicity, alone or in combination with other metric and geometric properties on graphs, can be used in the design of exact and approximate almost linear-time algorithms for the \textsc{Weighted Center} problem.
In particular, we derive almost optimal algorithms for the following well-studied classes of graphs: chordal graphs, distance-hereditary graphs (both in $\cO(m)$ time), dually chordal graphs and chordal bipartite graphs (both in $\cO(m\log{n})$ time).
\end{abstract}

\maketitle

\section{Introduction}\label{intro}

\textbf{Problem considered.} We refer to~\cite{BoMu08} for basics in Graph Theory. Unless stated otherwise, graphs considered throughout the paper will be simple, undirected, unweighted and connected. As usual, let $d_G$ denote the shortest-path metric of a graph $G=(V,E)$, \textit{i.e.}, for any vertices $u$ and $v$ their distance $d_G(u,v)$ is equal to the minimum number of edges on a $(u,v)$-path. The following facility location problem on graphs is considered:	

\defoptproblem{\textsc{Weighted Center}}{A graph $G = (V,E)$; a profile $\pi : V \mapsto \mathbb{R}_{\ge 0}$.}{A vertex $c \in V$ s.t. $r_\pi(c) = \max\{\pi(v)d_G(c,v) : v \in V\}$ is minimized.}

Roughly, $\pi(v)$ represents the importance for vertex $v$ to be close to the new facility.
For example, $\pi(v)$ can be the expected material flux between $v$ and a new factory, the number of children in a building at $v$ who will attend to a new neighborhood school, etc. 
Furthermore, if we set $\pi(v')=0$ for some vertices $v'$, then we can account for the plausible situation where the placement of a new facility is only relevant for a subset of the vertices.
Note that we cannot simply remove the zero-weight vertices for it could change the distances in the graph.

In line with existing literature~\cite{KaHa}, the function $r_\pi$ is called a \emph{radius function}. The solution vertex $c$ is called a \emph{center}. The \textsc{Weighted Center} problem can be defined more generally for any metric space. In particular, a linear-time algorithm for the \textsc{Weighted Center} problem on Euclidean spaces of constant dimension was presented in~\cite{Dyer}. For graphs, the existing literature is mostly focused on the following restricted variant:
		
\defoptproblem{\textsc{Center}}{A graph $G = (V,E)$.}{A vertex $c \in V$ s.t. $\ecc(c) = \max\{d_G(c,v) : v \in V\}$ is minimized.}

The radius function $\ecc$ is also called the \emph{eccentricity function} of $G$.
More generally, if $\pi(v) \in \{0,1\}$ for every vertex $v$, then $\pi$ is called a binary profile, and $r_\pi$ is called a binary radius function. The special case of binary profiles on graphs has received some attention in the literature, see~\cite{Cab,ChDrEsHaVa,dragan2026certificates}. As a rule of thumb, most of the techniques for the \textsc{Center} problem can be also applied for binary profiles. However, such is not the case for the more general \textsc{Weighted Center} problem. 

\medskip
\noindent
\textbf{Related work.} The \textsc{Weighted Center} problem can be reduced to All-Pairs Shortest-Paths. For $n$-vertex $m$-edge graphs, this straightforward reduction leads to an $\cO(nm)$-time algorithm. In~\cite{AVW16}, the Hitting Sets Conjecture was introduced to prove a conditional lower bound in $\Omega(n^{2-o(1)})$ for the \textsc{Center} problem on $n$-vertex graphs, even if the number of edges is at most in $n^{1+o(1)}$. Therefore, breaking the quadratic barrier for the \textsc{Center} problem, and even more so for the \textsc{Weighted Center} problem, is likely to require additional restrictions on the classes of graphs we consider. There is a plethora of subquadratic-time algorithms, and even better, of almost linear-time algorithms, for the \textsc{Center} problem on some well-structured classes of graphs; we refer to~\cite{AVW16,ChanVCdim,PFPT,DucCW} for a small sample of such results. By comparison, the only published linear-time algorithm for the \textsc{Weighted Center} problem, to the best of our knowledge, is that of Megiddo for the class of trees~\cite{Me2}. Recently, a parameterized almost linear-time algorithm for the \textsc{Weighted Center} problem on the graphs of bounded clique-width was presented in~\cite{DucCW}, using orthogonal range queries and the so-called convex hull trick. In particular, for the distance-hereditary graphs, the clique-width of which is at most three~\cite{GolRotCW}, this result leads to an $\cO(m\log^c{n})$-time algorithm for the \textsc{Weighted Center} problem, for some unspecified constant $c > 1$. An $\cO(m)$-time algorithm for the \textsc{Center} problem was presented much earlier, in~\cite{DrDH}. %The algorithm in~\cite{} mostly exploits a \emph{structural} graph property.
    
The authors in~\cite{Gpunimodal-ecc} introduced a different approach for the \textsc{Weighted Center} problem that is based on a \emph{metric} property, hereafter called \emph{$G^p$-unimodality} (see also~\cite{kHelly} for the case $p=1$). Roughly, a radius function $r_\pi$ on a graph $G$ is called $G^p$-unimodal if any vertex that minimizes $r_\pi$ within its ball of radius $p$ is a center. It was proved in~\cite{Gpunimodal-ecc} that for many classes of graphs studied in Metric Graph Theory, all their radius functions are $G^2$-unimodal. The latter was used in the design of randomized $\Tilde{\cO}(m\sqrt{n})$-time\footnote{The $\Tilde{\cO}()$ notation suppresses poly-logarithmic factors.} local-search algorithms for the \textsc{Weighted Center} problem on these classes of graphs. Although the algorithmic properties of the classes studied in Metric Graph Theory have insofar received limited attention (the results in~\cite{MedianOfMedians,BDHMedian,HellyHyp} count among the few exceptions), we want to stress here that they are natural generalizations of better-studied classes, including: chordal graphs, chordal bipartite graphs, and dually chordal graphs. In particular, before this work, the $\Tilde{\cO}(m\sqrt{n})$-time algorithms from~\cite{Gpunimodal-ecc,kHelly} were the best ones known for the \textsc{Weighted Center} problem on chordal (bipartite) graphs, and dually chordal graphs. By comparison, $\cO(m)$-time algorithms for the \textsc{Center} problem can be found in~\cite{BrChDrDuallyChordal,ChDrChordal,AbsoluteRetract}.
    
One common property of chordal graphs and their relatives, but also of distance-hereditary graphs, AT-free graphs, and several other well-studied classes of graphs (\textit{e.g.}, see~\cite{ChDrEsHaVa,WuZh01}) is that they have bounded \emph{Gromov hyperbolicity}. Roughly, the hyperbolicity of a graph represents how close, locally, its shortest-path metric is to a tree metric (see Sec.~\ref{prelim} for a formal definition). The introduction of this parameter by Gromov, see~\cite{Gr}, has revolutionized Geometric Group Theory. Later on, extensive experiments have evidenced that many complex networks also have bounded hyperbolicity~\cite{KenSanNa}. There is now a rich literature of approximation algorithms on graphs whose performances depend on the hyperbolicity value~\cite{HypAlg1,HypAlg2}. Of special interest to us is the main result from~\cite{ChDrEsHaVa}, which states that on any $\delta$-hyperbolic graph $G$, for every \emph{binary} profile $\pi$ on $G$, a vertex that is at distance at most $\cO(\delta)$ to the center can be computed in $\cO(m)$ time. Note that it leads to an approximation algorithm for the \emph{radius}: defined as $\rad_\pi(G) = \min\{ r_\pi(v) : v \in V \}$, with additive one-sided error in $\cO(\delta)$. Furthermore, this result was combined with the local-search techniques from~\cite{Gpunimodal-ecc} to refine the runtime of their exact algorithms (only for binary profiles) from $\cO(m\sqrt{n})$ to $\cO(\delta m)$. The motivation for our paper was to generalize the result from~\cite{ChDrEsHaVa} to arbitrary profiles.

\medskip
\noindent
\textbf{Our contributions.} For a graph $G$, and a profile $\pi$, let us define $C_\pi(G)$ as the set of all its centers. In what follows, the ball of center $v$ and radius $\rho$ is denoted by $B_\rho(v)$. We provide the following two main results for the \textsc{Weighted Center} problem on $\delta$-hyperbolic graphs.

\begin{thml}\label{th:hyp} Given a $\delta$-hyperbolic graph $G$ with $n$ vertices and $m$ edges, and an arbitrary profile $\pi$, one can compute:
	\begin{enumerate}
		\item\label{item:hyp-1} a vertex $c_1$ such that $C_\pi(G) \subseteq B_{12\delta\log{n}+4\delta+11}(c_1)$, in $\cO(m)$ time;
		\item\label{item:hyp-2} a vertex $c_2$ such that $C_\pi(G) \subseteq B_{124\delta + 21}(c_2)$, in $\cO(\delta m\log^2{n})$ time.
	\end{enumerate}
\end{thml}

In particular, for an \emph{arbitrary} profile $\pi$, a vertex that is at distance $\cO(\delta)$ to the center can be computed in $\Tilde{\cO}(\delta m)$ time.
This weakly generalizes the result of~\cite{ChDrEsHaVa} for binary profiles, but at the price of a slightly worse running time.
Furthermore, if $\pi$ is $k$-bounded, then we get an approximation of the radius with additive one-sided error in $\cO(k\delta)$.
The latter is also a multiplicative $(1+\cO(\varepsilon \delta))$-approximation of the radius, where $\varepsilon = \frac{\max_{v \in V}{\pi(v)}}{\rad_\pi(G)}$.

Theorem~\ref{th:hyp}(\ref{item:hyp-1}) follows from efficient algorithms for embedding a $\delta$-hyperbolic graph into a tree with additive distortion in $\cO(\delta \log{n})$, used in combination with the aforementioned linear-time algorithm for the \textsc{Weighted Center} problem on trees. By doing so, it is expected that we get a vertex $c_1$ such that $r_\pi(c_1)$ and $\rad_\pi(G)$ are close. However, it is less immediate that $c_1$ and the true centers of $G$ are also close. Theorem~\ref{th:hyp}(\ref{item:hyp-2}) is based on Theorem~\ref{th:hyp}(\ref{item:hyp-1}). Roughly, we search for $c_2$ in a ball of center $c_1$ and radius $\cO(\delta \log{n})$. For that, we reduce our search to $\Tilde{\cO}(\delta)$ instances of the \textsc{Weighted Center} problem with \emph{binary} profiles, which may be of independent interest. In particular, we are using the result of~\cite{ChDrEsHaVa}. 

\medskip
\noindent
As our second main contribution in the paper, we present a set of new algorithms for the \textsc{Weighted Center} problem on some important examples of graphs with bounded hyperbolicity.
These results are reported in Table~\ref{tab:graph-classes}. They are based on Theorem~\ref{th:hyp}, and on additional metric properties of these graphs. Our running times are optimal, or quasi optimal, for all the classes of graphs considered.
\begin{table}[!h]
    \centering
    \begin{tabular}{|c||c|c|c| }
        \hline
        Graph class & Prior works & \textbf{This work} & \textbf{Ref.} \\
        \hline
        \hline 
        Chordal & $\Tilde{\cO}(m\sqrt{n})$~\cite{Gpunimodal-ecc} & \textcolor{red}{$\cO(m)$} &  Theorem~\ref{thm:weakly-bridged}(\ref{item-chordal}) \\ 
        \hline 
        Chordal bipartite & $\Tilde{\cO}(m\sqrt{n})$~\cite{Gpunimodal-ecc} & \textcolor{red}{$\cO(m\log{n})$} & Theorem~\ref{th:helly}(\ref{item-chordal-like})  \\
        \hline 
        Distance-hereditary & $\cO(m\log^c{n}), \ c > 1$~\cite{DucCW} & \textcolor{red}{$\cO(m)$} & Theorem~\ref{thm:DH} \\
        \hline
        Dually chordal & $\Tilde{\cO}(m\sqrt{n})$~\cite{kHelly} & \textcolor{red}{$\cO(m\log{n})$} & Theorem~\ref{th:helly}(\ref{item-chordal-like}) \\
        \hline
    \end{tabular}
    \vspace{5pt}
    \caption{Results for the \textsc{Weighted Center} problem on some classes of graphs with bounded Gromov hyperbolicity.}
    \label{tab:graph-classes}
    \vspace{-5mm}
\end{table}

Our result for chordal graphs builds upon the local-search techniques from~\cite{Gpunimodal-ecc}.
Indeed, in this special case a vertex $c_1$ that is at distance $\cO(1)$ to the center can be computed in $\cO(m)$ time (the latter is a refinement of Theorem~\ref{th:hyp}(\ref{item:hyp-1}) for this class of graphs).
Then, starting from $c_1$, an improving sequence $(v_0=c_1,v_1,\ldots,v_\ell)$ can be computed in $\cO(\ell m)$ time, such that $r_\pi(v_i) > r_\pi(v_{i+1})$ for every $i$ such that $0 \le i < \ell$, and $v_\ell$ is a center.
If $\pi$ is a binary profile, then as our starting vertex $c_1$ is at distance $\cO(1)$ to the center, we must have $\ell = \cO(1)$ no matter how we choose $v_{i+1}$ from $v_i$.
However, as already noted in~\cite{Gpunimodal-ecc}, this does not hold anymore for arbitrary profiles. We introduce a new selection mechanism for the next improving vertex $v_{i+1}$ that ensures that the bound $\ell = \cO(1)$ still holds with arbitrary profiles. 

For distance-hereditary graphs, we use a similar local-search approach as for chordal graphs. For that, we prove that all their radius functions are also $G^2$-unimodal. We also design an $\cO(m)$-time procedure for computing the next improving vertex $v_{i+1}$ from $v_i$ in a distance-hereditary graph. The latter is based on \emph{injective hulls}, which are important objects in the study of metric spaces and their optimal realizations, see~\cite{Dress}. By doing so, we contribute to the recent line of studies of the properties of injective hulls of graphs~\cite{InjHullGraphs}.

For chordal bipartite graphs and dually chordal graphs, we adapt a technique from~\cite{HellyHyp}.
Roughly, it allows us to compute, in some ball of constant radius, all the vertices $v$ such that $r_\pi(v)$ is less than a fixed threshold $\tau$.
By doing so, all the centers can be computed with $\cO(\log{n})$ calls to this technique, using a nontrivial variation of binary search.
At this point, we should mention that it is not that obvious that we can efficiently reduce the \textsc{Weighted Center} problem to its decision version, using binary search, because the set of possible values for the radius is $\{ k \times \pi(v) : 0 < k < n, v \in V \}$, which has size in $\cO(n^2)$. To apply a naive binary-search strategy to this set would require first to sort all the values. The latter would result in an $\Omega(n^2\log{n})$-time lower bound. The existence of an efficient reduction from the \textsc{Weighted Center} problem to its decision version is a side contribution of the paper.

\medskip
\noindent
Our third contribution in the paper is a set of parameterized almost linear-time algorithms for the \textsc{Weighted Center} problem, where the parameter is the hyperbolicity value $\delta$.
These results are reported in Table~\ref{tab:param}.
\begin{table}[!h]
    \centering
    \begin{tabular}{|c||c|c| }
        \hline
        Graph class & Running time & Ref. \\
        \hline
        \hline 
        Max. degree $\Delta$ & $\Tilde{\cO}(\Delta^{\cO(\delta)}m + \delta m)$ & Theorem~\ref{th:growth}(\ref{item-max-degree})  \\ 
        \hline
        Weakly bridged & $\Tilde{\cO}(\delta m)$ & Theorem~\ref{thm:weakly-bridged}(\ref{item-wb-hyp}) \\ 
        \hline 
        Helly & $\Tilde{\cO}(\delta^3 m)$ &  Theorem~\ref{th:helly}(\ref{item-helly-gal}) \\
        \hline 
       Bipartite Helly & $\Tilde{\cO}(\delta^3 m)$ & Theorem~\ref{th:helly}(\ref{item-helly-gal})  \\
        \hline
        \textbf{Planar} & $\Tilde{\cO}(2^{O(\delta)}n)$ & Theorem~\ref{th:planar} \\
        \hline
    \end{tabular}
    \vspace{5pt}
    \caption{Parameterized algorithms for the \textsc{Weighted Center} problem on some classes of graphs, where the parameter is the hyperbolicity value $\delta$.}
    \label{tab:param}
    \vspace{-5mm}
\end{table}

Our result for graphs with bounded maximum degree is a straightforward application of Theorem~\ref{th:hyp}(\ref{item:hyp-2}).
For weakly bridged graphs, Helly graphs and bipartite Helly graphs, we use the same techniques as for chordal graphs, dually chordal graphs, and chordal bipartite graphs, of which they are superclasses.
The respective definitions of these classes of graphs are recalled in the technical sections of the paper.
Finally, our most interesting entry in Table~\ref{tab:param} is that for planar graphs, for which we do not use Theorem~\ref{th:hyp}.
Instead, we rely on a recent separator theorem for planar hyperbolic graphs, see~\cite{PlanarHyp}, which is combined with a \emph{VC-dimension} argument.
To the best of our knowledge, this is the first combination of Gromov hyperbolicity with VC-theory in algorithm design.

\medskip
\noindent
\textbf{Organization of the paper.} The notations and terminology that are used in the paper are first recalled in Sec.~\ref{prelim}.  Furthermore, we introduce some prior results that are required for the presentation of our algorithms and their analysis. Sec.~\ref{hyperbolic} is devoted to the proof of Theorem~\ref{th:hyp}. Then, in Sec.~\ref{exact}, we present all our exact algorithms for the \textsc{Weighted Center} problem on some classes of graphs with bounded hyperbolicity. Section~\ref{perspectives} concludes the paper, with some discussion about our work, and directions for future research.

\section{Preliminaries}\label{prelim}

We start recalling a few notations.
Let $G=(V,E)$ be a graph.
We will always denote by $n$ its number $|V|$ of vertices, and by $m$ its number $|E|$ of edges.
For every vertices $u$ and $v$, we write $u \sim v$ if $u$ and $v$ are adjacent, and $u \nsim v$ if they are nonadjacent.
Let $N_G(u) = \{ v \in V : u \sim v \}$ and $N_G[u] = N_G(u) \cup \{u\}$ be the (open) \emph{neighborhood} and the \emph{closed neighborhood} of vertex $u$.
For a subset $X$, let $N_G(X) = \bigcup\{ N_G(x) \setminus X : x \in X \}$.
Recall that a shortest $(u,v)$-path is one with minimum number of edges; the \emph{distance} $d_G(u,v)$ is equal to the length of a shortest $(u,v)$-path.
Furthermore, let $I_G(u,v) = \{ w \in V : d_G(u,v) = d_G(u,w) + d_G(w,v) \}$ be the (metric) \emph{interval} between $u$ and $v$.
Let also $I_G^o(u,v) = I_G(u,v) \setminus \{u,v\}$.
For every $k$ such that $0 \le k \le d_G(u,v)$, let $S_k(u,v,G) = \{ w \in I_G(u,v) : d_G(u,w) = k \}$ be called a \emph{slice}.
Recall that for every vertex $u$ and for every integer $\rho$, the set $B_\rho(u,G) = \{ v \in V : d_G(u,v) \le \rho \}$ is called the \emph{ball} of center $u$ and radius $\rho$.
For any vertex $u$, and for any subset $X$, let $d_G(u,X) = \min\{ d_G(u,x) : x \in X \}$.
%Then, for every integer $\rho$, we can also define $B_\rho(X,G) = \{ u \in V : d_G(u,X) \le \rho \}$.
The (weak) \emph{diameter} of a subset $X$ is defined as $\diam_G(X) = \max\{ d_G(x,x') : x,x' \in X \}$.
A subset $X$ is called (geodesically) \emph{convex} if $I_G(x,x') \subseteq X$ for any two vertices $x,x' \in X$. 
If $G$ is clear from the context, then it is omitted from our notations.
Other notations are locally defined at appropriate places in the paper.

\medskip
\noindent
\textbf{Functions defined on a graph.} Let $G=(V,E)$ be a graph, and let $f : V \mapsto \mathbb{R}$.
A \emph{global minimum} of $f$ is any vertex $u$ such that $f(u)$ is minimal within $V$.
A \emph{local minimum} of $f$ is any vertex $u$ such that $f(u) \le f(v)$ for every $v \in N(u)$. 
We call $f$ \emph{unimodal} if every local minimum is also a global minimum.
More generally, for every integer $p \ge 1$, let the \emph{$p^{th}$-power} of $G$ be the graph $G^p=(V,E_p)$ such that there exists an edge $uv \in E_p$ if and only if $0 < d_G(u,v) \le p$. 
We call $f$ \emph{$G^p$-unimodal} if it is unimodal on $G^p$.
Note that $f$ is $G^p$-unimodal if and only if every vertex $u$ that minimizes $f(u)$ within $B_p(u,G)$ is also a global minimum.
The function $f$ is called \emph{weakly peakless} if it satisfies the following condition $\WP(u,v)$ for every distinct vertices $u$ and $v$ such that $u \nsim v$: there exists a vertex $w \in I^o(u,v)$ such that $f(w) \le \max\{f(u),f(v)\}$, and equality holds only if $f(u) = f(w) = f(v)$. It is called \emph{$p$-weakly peakless} if it satisfies $\WP(u,v)$ for every vertices $u$ and $v$ such that $d(u,v) \ge p+1$.

\begin{mylemma}[see Theorem $11$ in~\cite{WP}]\label{lem:wp}
    Every $p$-weakly peakless function on a graph $G$ is $G^p$-unimodal.
\end{mylemma}

\medskip
\noindent
\textbf{Results on the \textsc{Weighted Center} problem.} Let $G=(V,E)$ be a graph, and let $\pi : V \mapsto \mathbb{R}_{\ge 0}$ be a profile. 
The support of $\pi$, hereafter denoted by $\Supp(\pi)$, is the set of all vertices $v$ such that $\pi(v) > 0$.
Recall that $r_\pi : u \in V \mapsto \max\{ \pi(v)d(u,v) : v \in V \}$ is called a radius function.
Furthermore, for every vertex $u$, let $F_\pi(u) = \{ v \in V : r_\pi(u) = \pi(v)d(u,v) \}$.
Recall that we define the ($\pi$-)radius as $\rad_\pi(G) = \min\{ r_\pi(u) : u \in V \}$.
In the same way, let the center $C_\pi(G)$ contain all the vertices $c$ such that $r_\pi(c)$ is minimized; each vertex $c \in C_\pi(G)$ is also called a center.
As a side contribution of the paper, we present a reduction from the \textsc{Weighted Center} problem to its decision version.
For that, although the search space for $\rad_\pi(G)$ is in $\cO(n^2)$, we can use a strategy inspired by the median-of-medians algorithm~\cite{MedianOfMediansArray} to quickly reduce the size of this space.
\begin{mylemma}\label{lem:weighted-center-decision}
    Let $G=(V,E)$ be a graph, and let $\pi$ be an arbitrary profile on $G$.
    If the problem of either computing a vertex $c$ such that $r_\pi(c) \le \tau$, or asserting that no such vertex exists, can be solved in $\cO(T(n,m))$ time for any value $\tau$, then the \textsc{Weighted Center} problem can be solved in $\cO\left((T(n,m)+n)\log{n}\right)$ time on $G,\pi$.     
\end{mylemma}
\begin{proof}
    If $\Supp(\pi) = \emptyset$, then $C_\pi(G) = V$. Furthermore, if $\Supp(\pi)$ is a singleton, then $C_\pi(G) = \Supp(\pi)$.
    Thus, let us assume from now on $|\Supp(\pi)| > 1$.
    We maintain a set $S \subseteq V \times \mathbb{N} \times \mathbb{N}$ with the property that for some $(v,a_v,b_v) \in S$, there exists some $k$ such that $a_v \le k \le b_v$, and $\rad_\pi(G) = \pi(v)k$.
    Initially, we set $S = \{ (v,1,n-1) : v \in \Supp(\pi) \}$.
    Then, our algorithm iteratively performs some procedure, detailed next, until the following halting condition, (HC), is met: for every $(v,a_v,b_v) \in S$, $a_v = b_v$.
    Note that once (HC) holds, we are left with $\cO(|S|) = \cO(n)$ distinct possibilities for $\rad_\pi(G)$.
    We can then sort all these values in $\cO(n\log{n})$ time, remove duplicates, and compute a center in $\cO(T(n,m)\log{n})$ time by using a standard binary search.

    Assume for what follows that (HC) does not hold.
    For every $v \in \Supp(\pi)$, we define its weight $\omega(v)$ as follows: if there is some triple $(v,a_v,b_v)$ in $S$, then $\omega(v) = b_v-a_v$, and otherwise $\omega(v) = 0$.
    Furthermore, let $\omega(S) = \sum\{ \omega(v) : v \in \Supp(\pi) \}$.
    As (HC) does not hold, $\omega(S) > 0$.
    Then, the following procedure is performed:
    \begin{enumerate}
        \item For some arbitrary ordering $v_0,v_1,\ldots,v_\ell$ of all the vertices of positive weight, we construct an array whose entries are the values $\tau_i := \pi(v_i)\left\lfloor\frac{a_{v_i}+b_{v_i}}{2}\right\rfloor$, for $0 \le i \le \ell$.
        \item We compute the \emph{lowest weighted median} of the array: defined as the value $\tau_{i^*}$ such that $\sum\{ \omega(v_j) : \tau_j < \tau_{i^*} \} < \omega(S)/2$ and $\sum\{ \omega(v_j) : \tau_j > \tau_{i^*} \} \le \omega(S)/2$.
        \item We test whether there exists a vertex $c$ such that $r_\pi(c) \le \tau_{i^*}$.
            \begin{enumerate}
                \item If yes, then for each $(v,a_v,b_v) \in S$, let us define $b_v' = \left\lfloor \frac{\tau_{i^*}}{\pi(v)} \right\rfloor$.
                If $b_v' < a_v$, then we remove this triple from $S$. Otherwise, if in addition $b_v' < b_v$, then we replace $(v,a_v,b_v)$ with $(v,a_v,b_v')$ in $S$.
                \item Else, for each $(v,a_v,b_v) \in S$, let us define $a_v' = \left\lfloor \frac{\tau_{i^*}}{\pi(v)} \right\rfloor + 1$.
                If $a_v' > b_v$, then we remove this triple from $S$. Otherwise, if in addition $a_v' > a_v$, then we replace $(v,a_v,b_v)$ with $(v,a_v',b_v)$ in $S$.
            \end{enumerate}
    \end{enumerate}
    As noted in~\cite{Cormen}, the lowest weighted median can be computed in $\cO(\ell) = \cO(n)$ time by using a variation of the median-of-medians algorithm, see~\cite{MedianOfMediansArray}. 
    Therefore, this above procedure can be done in $\cO(T(n,m)+n)$ time.
    Furthermore, we claim that the total weight $\omega(S)$ is decreased by a factor at least $3/4$ after each call to the procedure.
    Indeed, if there is a vertex $c$ such that $r_\pi(c) \le \tau_{i^*}$, then the weight of every vertex $v_j$ such that $\tau_j \ge \tau_{i^*}$ is at least halved.
    If there is no such vertex, then the weight of every vertex $v_j$ such that $\tau_j \le \tau_{i^*}$ is at least halved.
    Therefore, we always halve the weights for a subset of vertices such that the sum of their weights is at least $\omega(S)/2$.
    In particular, the total weight decreases by at least $\omega(S)/4$, thus proving our claim.
    Overall, as initially $\omega(S) = \cO(n^2)$, there are $\cO(\log{n})$ calls to the procedure.
\end{proof}

\medskip
\noindent
\textbf{$\delta$-hyperbolic graphs and their properties.} 
For every nodes $u,v,x,y$ in a tree $T$, it holds $d_T(u,v)+d_T(x,y) \le \max\{d_T(u,x)+d_T(v,y),d_T(u,y)+d_T(v,x)\}$.
The definition of Gromov hyperbolicity can be regarded as a relaxation of this four-point condition.
A graph $G=(V,E)$ is called \emph{$\delta$-hyperbolic} if for every $u,v,x,y \in V$, $d(u,v)+d(x,y) \le \max\{d(u,x)+d(v,y),d(u,y)+d(v,x)\} + 2\delta$ holds.
Alternatively, if $d(u,y)+d(v,x) \le d(u,x)+d(v,y) \le d(u,v)+d(x,y)$, then $d(u,v)+d(x,y) - d(u,x)-d(v,y) \le 2\delta$.
The \emph{hyperbolicity} of $G$ is the smallest $\delta$ such that it is $\delta$-hyperbolic.
In what follows, we gather some properties of $\delta$-hyperbolic graphs. 
%We will also need the following Helly-type property for $\delta$-hyperbolic graphs:

The first result is a relaxation of a well-known Helly property for trees: every family of pairwise intersecting subtrees (and so, every family of pairwise intersecting balls in a tree) has a nonempty common intersection. 
\begin{mylemma}[see Lemma $1$ in~\cite{ChEs}]\label{lem:hyp-helly-gap}
    If $G$ is $\delta$-hyperbolic, then for any family of pairwise intersecting balls $B_{\rho_1}(v_1), B_{\rho_2}(v_2), \ldots, B_{\rho_k}(v_k)$, $\bigcap\{ B_{\rho_i+2\delta}(v_i) : 1 \le i \le k \} \ne \emptyset$.
\end{mylemma}
The next result is a special case of Morse lemma~\cite{Sh}, of which we give a direct proof.
Recall that in a tree, there exists a unique path between every two nodes.
Roughly, we relax this property for a $\delta$-hyperbolic graph such that almost shortest-paths between every two vertices must stay close to each other.
\begin{mylemma}\label{lem:hyp-morse}
    Let $u,v,x,y$ be vertices in a $\delta$-hyperbolic graph $G$.
    If $u \in I(x,y)$, $d(x,v)+d(v,y) \le d(x,y)+\lambda$, and $d(v,x) \ge d(u,x), \ d(v,y) \ge d(u,y)$, then $d(u,v) \le \lambda + 2\delta$.
\end{mylemma}
\begin{proof}
    Since $G$ is $\delta$-hyperbolic, $d(u,v) + d(x,y) \le \max\{d(u,x)+d(v,y),d(u,y)+d(v,x)\} + 2\delta$.
    Therefore, $d(u,v)+d(x,y) \le \max\{d(u,x),d(v,x)\} + \max\{d(u,y),d(v,y)\} + 2\delta \le d(v,x) + d(v,y) + 2\delta \le d(x,y) + \lambda + 2\delta$. The latter implies as claimed that $d(u,v) \le \lambda + 2\delta$.
\end{proof}
Finally, observe that for any subtree $H$ in a tree (in particular, for any ball in a tree), the nodes out of $H$ can be partitioned in rooted subtrees with their roots being taken in the boundary of $H$. 
In particular, every node $y$ out of $H$ has a unique closest node $v$ in $H$ (namely, the root of its subtree), and the unique path from $y$ to any node of $H$ goes by $v$. 
We prove a relaxation of this simple property for balls in a $\delta$-hyperbolic graph, namely:
\begin{mylemma}\label{lem:hyp-proj}
    Let $u,x,y$ be vertices in a $\delta$-hyperbolic graph $G$.
    If $x \in B_r(u)$ and $y \notin B_r(u)$, then for any $v \in S_r(u,y)$, $d(x,v) + d(v,y) \le d(x,y) + 2\delta$.
\end{mylemma}
\begin{proof}
    Consider the three sums $A = d(u,y)+d(v,x)$, $B = d(u,v)+d(x,y)$, and $C = d(u,x)+d(v,y)$. 
	Since $v \in S_r(u,y)$, we get that $d(u,v) = r \ge d(u,x)$.
    Furthermore, $d(y,z) \ge d(y,v)$ for any $z \in B_r(u)$.
    So, in particular, $d(x,y) \ge d(v,y)$. 	
	Therefore, $B \ge C$.
    As $v \in I(u,y)$, $d(u,y) = d(u,v) + d(v,y)$. Hence, $A = d(u,v) + d(v,y) + d(v,x) \ge d(u,v) + d(x,y) = B$.
    Since $G$ is $\delta$-hyperbolic, $A - B \le 2\delta$.
    As a result, $d(v,y) + d(v,x) - d(x,y) \le 2\delta$.
\end{proof}

\section{Approximation algorithms for $\delta$-hyperbolic graphs}\label{hyperbolic}

This section is devoted to the proof of Theorem~\ref{th:hyp}.
Its statements (\ref{item:hyp-1}) and (\ref{item:hyp-2}) are proved in Sec.~\ref{hyp-weak} and Sec.~\ref{hyp-strong}, respectively.

\subsection{Using tree embeddings}\label{hyp-weak}
A graph $G$ \emph{$\lambda$-embeds} into a tree if there exists a tree $T$, possibly edge-weighted, such that $V \subseteq V(T)$ and for every vertices $u$ and $v$ of $G$, $d_G(u,v) \le d_T(u,v) \le \lambda d_G(u,v)$.
This part is devoted to the proof of the following result:	
	
\begin{mytheorem}\label{thm:hyp-weak}
	If a $\delta$-hyperbolic graph $G=(V,E)$ $\lambda$-embeds into a tree, then for any profile $\pi$, one can output in $\cO(m)$ time a vertex $c_1$ such that $C_\pi(G) \subseteq B_{6\lambda+4\delta+5}(c_1)$.
\end{mytheorem}		

Note that Theorem~\ref{th:hyp}(\ref{item:hyp-1}) follows from Theorem~\ref{thm:hyp-weak} and the following Lemma~\ref{lem:lambda-hyp}.
\begin{mylemma}[see Proposition $6.1$.B in~\cite{Gr}]\label{lem:lambda-hyp}
    Every $\delta$-hyperbolic graph $\lambda$-embeds into a tree, for some $\lambda \le 2\delta\log{n}+1$.
\end{mylemma}
Stronger versions of Lemma~\ref{lem:lambda-hyp} are used for some of the classes of graphs considered in Sec.~\ref{exact}.
		
Roughly, the proof of Theorem~\ref{thm:hyp-weak} consists in solving approximately the \textsc{Weighted Center} problem in a $\delta$-hyperbolic graph $G$ by solving this problem exactly on a tree $T$ in which $G$ embeds. It is known from prior works that such a tree $T$ with almost optimal (additive) distortion can be computed efficiently, namely:	
		
\begin{mylemma}[see Corollary 4 in~\cite{ChDrNeRaVa}]\label{lem:tree-embedding}
	If a graph $G=(V,E)$ $\lambda$-embeds into a tree, then there exists an unweighted tree $T=(V,F)$ (without Steiner points), constructible in $\cO(m)$ time, such that $d_T(x,y)-2 \le d_G(x,y) \le d_T(x,y) + 3\lambda$ for any vertices $x,y \in V$.
\end{mylemma}
	
Given $G$ and $T$ as above, our main contribution in this part is to prove that for an arbitrary profile $\pi$, for any centers $c \in C_\pi(G)$ and $c' \in C_\pi(T)$, their distance (in $G$) is at most linear in the hyperbolicity of $G$ and the additive distortion of the embedding. We prove it next:
	
\begin{myproposition}\label{prop:approx-center-hyp}
Given a $\delta$-hyperbolic graph $G=(V,E)$ and a tree $T=(V,F)$ such that $d_T(x,y) - p \le d_G(x,y) \le d_T(x,y) + q$ for any vertices $x,y \in V$ ($p,q \ge 0$), for an arbitrary profile $\pi$, $C_\pi(G) \subseteq B_{2(p+q)+4\delta+1}(c')$ holds for any $c' \in C_\pi(T)$.
\end{myproposition}
\begin{proof}
	Set $\tau = p+q+2\delta+1$. Suppose for the sake of contradiction the existence of a center $c \in C_\pi(G)$ such that $d_G(c,c') \ge 2\tau$. Let $w \in I_G(c,c')$ such that $d_G(c,w),d_G(w,c') \ge \tau$. 
	
	We consider some vertex $x$ such that $\pi(x)d_G(w,x)$ is maximized. Since $G$ is $\delta$-hyperbolic, $d_G(c,c')+d_G(w,x) \le \max\{d_G(c,w)+d_G(c',x),d_G(c,x)+d_G(c',w)\}+2\delta$. By the choice of vertex $w$, $d_G(c,c') \ge \max\{d_G(c,w),d_G(c',w)\} + \tau$, and so, $d_G(w,x) \le \max\{d_G(c',x),d_G(c,x)\}+2\delta - \tau$. Furthermore, by the hypothesis $2\delta-\tau < 0$. Since $d_G(c,x) \le d_G(w,x)$ because $c \in C_\pi(G)$, it follows that $d_G(w,x) \le d_G(c',x) + 2\delta-\tau$. 
	
	We also consider some vertex $y$ such that $\pi(y)d_T(w,y)$ is maximized. Again since $G$ is $\delta$-hyperbolic, $d_G(c,c')+d_G(w,y) \le \max\{d_G(c,w)+d_G(c',y),d_G(c,y)+d_G(c',w)\}+2\delta$. Recall that by the choice of vertex $w$, $d_G(c,c') \ge \max\{d_G(c,w),d_G(c',w)\} + \tau$, and so, $d_G(w,y) \le \max\{d_G(c',y),d_G(c,y)\}+2\delta - \tau$. Since $c' \in C_\pi(T)$, we obtain $d_T(c',y) \le d_T(w,y)$, and so, $d_G(c',y) \le d_T(c',y) + q \le d_T(w,y) + q \le d_G(w,y) + p + q$. Furthermore, by the hypothesis $p+q+2\delta-\tau < 0$. Therefore, $d_G(w,y) \le d_G(c,y) + 2\delta - \tau$. 
	
	Suppose $\pi(x)d_G(w,x) \le \pi(y)d_T(w,y)$. Then, $d_T(w,y) \le d_G(w,y) +p \le d_G(c,y) + p + 2\delta - \tau < d_G(c,y)$. It implies $\pi(y)d_G(c,y) > \pi(y)d_T(w,y) \ge \pi(x)d_G(w,x)$, thus contradicting our assumption that $c \in C_\pi(G)$. Therefore, $\pi(x)d_G(w,x) > \pi(y)d_T(w,y)$. Then, $d_G(w,x) \le d_G(c',x) + 2\delta - \tau \le d_T(c',x)+q+2\delta-\tau < d_T(c',x)$. However, it implies $\pi(x)d_T(c',x) > \pi(x)d_G(w,x) > \pi(y)d_T(w,y)$, thus contradicting our assumption that $c' \in C_\pi(T)$.
\end{proof}
	
	Theorem~\ref{thm:hyp-weak} now follows from the combination of Lemma~\ref{lem:lambda-hyp}, Proposition~\ref{prop:approx-center-hyp}, and the following known algorithmic result:	
	
\begin{mylemma}[see Section 3 in~\cite{Me2}]
	Given a tree $T=(V,E)$, and an arbitrary profile $\pi$, one can solve the \textsc{Weighted Center} problem in $\cO(n)$ time.
\end{mylemma}

\subsection{Getting closer to the center}\label{hyp-strong}
	
	The main result of this section is as follows:

\begin{mytheorem}\label{thm:hyp-strong}
	Given a $\delta$-hyperbolic graph $G=(V,E)$, an arbitrary profile $\pi$, and a vertex $c_1$ and a radius $\rho$ such that $C_\pi(G) \subseteq B_\rho(c_1)$, one can output in $\cO(\rho m\log{n})$ time a vertex $c_2$ such that $C_\pi(G) \subseteq B_{124\delta +21}(c_2)$.
\end{mytheorem}
Note that Theorem~\ref{thm:hyp-strong}, in combination with Theorem~\ref{th:hyp}(\ref{item:hyp-1}) implies Theorem~\ref{th:hyp}(\ref{item:hyp-2}).
To prove Theorem~\ref{thm:hyp-strong}, roughly, our strategy consists in projecting the support of the profile onto the ball $B_\rho(c_1)$. Then, roughly, we reduce the computation of an approximate center $c_2$ for $\pi$ to that of approximate centers for $\cO(\rho \log{n})$ binary profiles.
For that, we will also need some results on the \textsc{Weighted Center} problem for binary profiles:

\begin{mylemma}[see Proposition $6$ in~\cite{ChDrEsHaVa}]\label{lem:approx-center}
    Given a $\delta$-hyperbolic graph $G=(V,E)$, and a binary profile $\pi$, one can output in $\cO(m)$ time a vertex $c^*$ such that $C_\pi(G) \subseteq B_{5\delta+1}(c^*)$.
\end{mylemma}

Consider the important special case of the eccentricity function, \textit{i.e.}, where $\pi(v) = 1$ for every vertex $v$.
Then, it is known that for any node $v$ in a tree $T$, $\ecc(v) = d(v,C(T))+\rad(T)$.
Furthermore, a relaxation of this result was proved in~\cite{DrGu_hyp}, Theorem 5, for all $\delta$-hyperbolic graphs. 
We prove it next for any binary profile.
    
\begin{mylemma}\label{lem:almost-unimodal}
    Let $G=(V,E)$ be $\delta$-hyperbolic, and let $\pi$ be a binary profile on $G$.
    For any vertex $v \in V$, $d(v,C_\pi(G)) + \rad_\pi(G) - 6\delta - 1 \le r_\pi(v) \le d(v,C_\pi(G)) + \rad_\pi(G)$.
\end{mylemma}
\begin{proof}
	The upper bound holds for any graph $G$ and any binary profile. For the lower bound, let $c \in C_\pi(G)$ be arbitrary, and let $u,w \in \Supp(\pi)$ be maximizing $d(u,w)$. Since $G$ is $\delta$-hyperbolic, $d(v,c) + d(u,w) \le \max\{d(v,u)+d(c,w),d(v,w)+d(c,u)\} + 2\delta$, and so, $d(v,c) + d(u,w) \le r_\pi(v) + r_\pi(c) + 2\delta = r_\pi(v) + \rad_\pi(G) + 2\delta$. By~\cite{ChDrEsHaVa}, Proposition 4, $d(u,w) \ge 2 \rad_\pi(G) - 4\delta-1$. Therefore, $r_\pi(v) \ge d(v,c) + d(u,w) - \rad_\pi(G) - 2\delta \ge d(v,C_\pi(G)) + d(u,w) - \rad_\pi(G) - 2\delta \ge d(v,C_\pi(G)) + (2 \rad_\pi(G) - 4\delta-1) - \rad_\pi(G) - 2\delta = d(v,C_\pi(G)) + \rad_\pi(G) - 6\delta -1$.
\end{proof}

For simplicity of the presentation, we assume in what follows that the hyperbolicity value $\delta$ is known to us.
In practice, any approximation $\delta' \ge \delta$ can be used, but then the distance of $c_2$ to $C_\pi(G)$ also depends on the approximation ratio.
\begin{proof}[Proof of Theorem~\ref{thm:hyp-strong}]
    We shall assume throughout the proof that $\rho > 124 \delta + 21$, since otherwise it suffices to output $c_2 = c_1$.
    Furthermore, we shall also assume that $\Supp(\pi)$ contains at least two vertices: otherwise, either $\Supp(\pi) = \emptyset$, and so, $C_\pi(G) = V$ holds; or $\Supp(\pi) = \{v\}$ for some $v \in V$, and then, $C_\pi(G) = \{v\}$ holds.    
    By Lemma~\ref{lem:weighted-center-decision} (slightly modified), we only need to present an $\cO(\rho m)$-time algorithm such that, for any value $\tau$: it either asserts $\rad_{\pi}(G) < \tau$; or it asserts $\rad_{\pi}(G) > \tau$; or it outputs some $c^*$ such that $C_{\pi}(G) \subseteq B_{124\delta+21}(c^*)$.

    For every $v \in \Supp(\pi)$, let $q(v)$ be a closest-to-$v$ vertex in the ball $B_\rho(c_1)$, and let $\ell(v) = d(v,q(v))$.
    -- Roughly, we project the support of $\pi$ onto $B_\rho(c_1)$. --
    So, in particular, if $d(c_1,v) \le \rho$, then we set $q(v) = v$, and $\ell(v) = 0$.
    If $d(c_1,v) > \rho$, then we set $\ell(v) = d(v,c_1)-\rho$, and we choose an arbitrary $q(v) \in S_\rho(c_1,v)$.
    All the values $q(v),\ell(v)$, for $v \in \Supp(\pi)$, can be computed by dynamic programming on a shortest-path tree rooted at $c_1$.

    Let $\tau$ be some value to be compared to $\rad_\pi(G)$.
    For each $v \in \Supp(\pi)$, we compute $\lambda_\tau(v) = \left\lfloor \frac{\tau}{\pi(v)} \right\rfloor$: the maximum distance between $v$ and any $c$ such that $r_\pi(c) \le \tau$. If $\lambda_\tau(v) < \ell(v)$ for some $v \in \Supp(\pi)$, then as $C_\pi(G) \subseteq B_\rho(c_1)$, we assert $\rad_{\pi}(G) > \tau$. If $\lambda_\tau(v) \ge \ell(v)+2\rho$, then as $C_\pi(G) \subseteq B_\rho(c_1) \subseteq B_{\lambda_\tau(v)}(v)$, we can discard $v$. Furthermore, if $\lambda_\tau(v') \ge \ell(v')+2\rho$ for \emph{every} $v' \in \Supp(\pi)$, then it implies $r_\pi(c_1) < \tau$, and so, we assert $\rad_{\pi}(G) < \tau$. See Claim~\ref{claim:large-lambda}. Thus, we now assume $\ell(v) \le \lambda_\tau(v) \le \ell(v)+2\rho-1$ for each $v \in \Supp(\pi)$. To avoid some degenerate cases in our analyzes, we further assume $\lambda_\tau(v) > 0$ for each $v \in \Supp(\pi)$. Indeed, assume $\lambda_\tau(v) = 0$ for some $v$. If $v \notin B_\rho(c_1)$, then, as $\ell(v) > \lambda_\tau(v)$, we assert $\rad_\pi(G) > \tau$. Otherwise, either $r_\pi(v) \le \tau$, or $r_\pi(v) > \tau$. In the former case, necessarily $C_\pi(G) = \{v\}$. In the latter case, we assert $\rad_\pi(G) > \tau$.

    For every $i$ such that $0 \le i \le 2\rho-1$, we consider $M_i = \{ q(v) : \lambda_\tau(v) = \ell(v)+i\}$.
    Roughly, if $M_i \ne \emptyset$, then we want to construct some superset $C^i$ of $C_\pi(G)$. 
    More specifically, let us define a binary profile $\pi_i$ such that $\pi_i(v) = 1$ if and only if $v \in M_i$.
    By Lemma~\ref{lem:approx-center}, a vertex $c^i$ such that $C_{\pi_i}(G) \subseteq B_{5\delta+1}(c^i)$ can be computed in $\cO(m)$ time.
    If $r_{\pi_i}(c^i) > i + 7\delta+1$, then we assert $\rad_\pi(G) > \tau$. 
    Roughly, this is because if $\rad_\pi(G) \le \tau$, then we can prove $\rad_{\pi_i}(G) \le i+2\delta$ (and so, $r_{\pi_i}(c^i) \le i + 7\delta+1$) using Lemma~\ref{lem:hyp-proj}. See Claim~\ref{claim:rad-ci}.
    Otherwise, we set $C^i = B_{i-r_{\pi_i}(c^i)+18\delta +3}(c^i) \cap B_\rho(c_1)$. 

    Let $C^* = \bigcap\{ C^i : M_i \ne \emptyset  \}$. 
    If $C^* = \emptyset$, then we assert $\rad_{\pi}(G) > \tau$.
    Roughly, this is because if $\rad_\pi(G) \le \tau$, then we can use Lemma~\ref{lem:almost-unimodal} to prove $C_\pi(G) \subseteq C^*$.
    See Claim~\ref{claim-above-alpha}.
    Else, let $c^* \in C^*$ be arbitrary. 
    If $C^* \subseteq B_{124\delta + 21}(c^*)$, then we output $c^*$.
    Otherwise, we assert that $\rad_{\pi}(G) < \tau$.
    This is the most technical part of the proof.
    Recall that if $\rad_\pi(G) \le \tau$, then we can prove $C_\pi(G) \subseteq C^*$.
    So, we need to consider the case $\rad_\pi(G) \ge \tau$.
    Then, we prove that for some $u,v \in \Supp(\pi)$, $\lambda_\tau(u)+\lambda_\tau(v)- \cO(\delta) \le d(u,v) \le \lambda_\tau(u)+\lambda_\tau(v) + \cO(\delta)$.
    For that, we exploit both the fact that $C^* \ne \emptyset$ (for the upper bound) and the Helly-type property of Lemma~\ref{lem:hyp-helly-gap} (for the lower bound). 
    Furthermore, due the way we construct $C^*$, and due to the definition of $C_\pi(G)$, we observe that every vertex of $C^* \cup C_\pi(G)$ must be contained in an almost shortest $(u,v)$-path: at distances $\lambda_\tau(u) \pm \cO(\delta)$ and $\lambda_\tau(v) \pm \cO(\delta)$ of $u$ and $v$. As a result, we can use Lemma~\ref{lem:hyp-morse} to upper bound the diameter of $C^* \cup C_\pi(G)$. See Claims~\ref{claim-below-alpha} and~\ref{claim-close-centre} in what follows. 

    As we call Lemma~\ref{lem:approx-center} $\cO(\rho)$ times, the runtime is in $\cO(\rho m)$.

    \medskip
    \noindent
    {\bf Correctness.} The following claims complete the correctness proof of the algorithm.
    We put all the claims together at the end of the proof. 

    \begin{myclaim}\label{claim:large-lambda}
    If $\lambda_\tau(v) \ge \ell(v)+2\rho$ for every $v \in \Supp(\pi)$, then $\rad_\pi(G) < \tau$.    
    \end{myclaim}
    Indeed, for any $v \in \Supp(\pi)$, $d(c_1,v) \le \ell(v)+\rho \le \ell(v)+2\rho-1$.
    So, in particular, $\bigcap\{ B_{\lambda_\tau(v)-1}(v) : v \in \Supp(\pi) \} \ne \emptyset$. $\diamond$

    \begin{myclaim}\label{claim:rad-ci}
    If $M_i \ne \emptyset$, and $z$ is any vertex of $B_\rho(c_1)$ such that $r_\pi(z) \le \tau$, then $r_{\pi_i}(z) \le i+2\delta$.
    Furthermore, if $\rad_\pi(G) \le \tau$, then $\rad_{\pi_i}(G) \le i+2\delta$.  
    \end{myclaim}
    Suppose, for the sake of contradiction, $r_{\pi_i}(z) > i+2\delta$.
    Let $v \in \Supp(\pi)$ be chosen such that $q(v) \in F_{\pi_i}(z)$.
    Assume first that $v \in B_\rho(c_1)$.
    As $d(z,v) > i = \lambda_\tau(v)$, we get $r_\pi(z) \ge \pi(v)d(z,v) > \tau$, thus contradicting our assumption that $r_\pi(z) \le \tau$.
    Thus, $v \notin B_\rho(c_1)$.
    As we assume $z \in B_\rho(c_1)$, by Lemma~\ref{lem:hyp-proj}, we obtain $d(v,q(v))+d(q(v),z) \le d(v,z)+2\delta$.
    Therefore, $d(v,z) \ge \ell(v)+d(q(v),z)-2\delta > \ell(v)+i = \lambda_\tau(v)$.
    Again, this contradicts our assumption that $r_\pi(z) \le \tau$.
    So, $\rad_{\pi_i}(G) \le r_{\pi_i}(z) \le i+2\delta$. 
    Finally, note that if $\rad_\pi(G) \le \tau$, then as we assume $C_\pi(G) \subseteq B_\rho(c_1)$, there exists a $z \in B_\rho(c_1)$ such that $r_\pi(z) \le \tau$. $\diamond$ 

    \begin{myclaim}\label{claim-above-alpha}
    For any vertex $z \in B_\rho(c_1)$ such that $r_{\pi}(z) \le \tau$, $z \in C^*$.
    \end{myclaim}
   It suffices to prove that $z \in C^i$ for any $i$ such that $M_i \ne \emptyset$. 
    By Claim~\ref{claim:rad-ci}, $r_{\pi_i}(z) \le i+2\delta$. 
    Therefore, by Lemma~\ref{lem:almost-unimodal}, $d(z,C_{\pi_i}(G)) \le r_{\pi_i}(z) - \rad_{\pi_i}(G) + 6\delta + 1 \le i - \rad_{\pi_i}(G) + 8\delta+1$.
    As $C_{\pi_i}(G) \subseteq B_{5\delta+1}(c^i)$, it implies that
    \begin{flalign*}
        d(z,c^i) &\le d(z,C_{\pi_i}(G)) + 5\delta + 1 \le i - \rad_{\pi_i}(G) + 13\delta+2 \\
        &\le i - (r_{\pi_i}(c^i)-5\delta-1) + 13\delta+2 \\
        &\le i - r_{\pi_i}(c^i) + 18\delta+3.
    \end{flalign*}
    Therefore, $z \in C^i$. $\diamond$

    \begin{myclaim}\label{claim:dist-cstar-support}
    If $z \in C^*$, and $v \in \Supp(\pi)$, then $d(z,v) \le \lambda_\tau(v)+18\delta+3$.    
    \end{myclaim}
    As $C^* \ne \emptyset$, we get $\lambda_\tau(v) > 0$.
    If $\lambda_\tau(v) \ge \ell(v)+2\rho$, then it holds $C^* \subseteq B_\rho(c_1) \subseteq B_{\lambda_\tau(v)}(v)$, and so, we are done.
    Thus, from now on, let $i$ be such that $0 \le i \le 2\rho-1$, and $q(v) \in M_i$.
    Then, $d(z,v) \le d(z,c^i)+d(c^i,q(v))+d(q(v),v)$.
    Recall that $d(q(v),v) = \ell(v) = \lambda_\tau(v)-i$.
    Furthermore, as $q(v) \in M_i$, $d(c^i,q(v)) \le r_{\pi_i}(c^i)$.
    As $z \in C^* \subseteq C^i$, we get $d(z,c^i) \le i - r_{\pi_i}(c^i) + 18\delta+3$.
    As a result,
    \begin{flalign*}
        d(z,v) &\le (i - r_{\pi_i}(c^i) + 18\delta+3) + r_{\pi_i}(c^i) + (\lambda_\tau(v)-i) \\
        &\le \lambda_\tau(v) + 18\delta+3. \ \ \diamond
    \end{flalign*}

    \begin{myclaim}\label{claim-below-alpha}
    If $\rad_{\pi}(G) \ge \tau$, and if $C^* \ne \emptyset$, then $\diam(C^*) \le 124\delta +21$.
    \end{myclaim}
    Assume first $\lambda_\tau(v) \le 2\delta+1$ for some $v \in \Supp(\pi)$.
    By Claim~\ref{claim:dist-cstar-support}, $C^* \subseteq B_{20\delta+4}(v)$ holds.
    Therefore, it holds $\diam(C^*) \le 40\delta+8$.
    From now on, we assume that $\lambda_\tau(v) > 2\delta+1$ for every $v \in \Supp(\pi)$.
    As we assume $\rad_{\pi}(G) \ge \tau$, we obtain $\bigcap\{ B_{\lambda_\tau(v)-1}(v) : v \in \Supp(\pi) \} = \emptyset$.
    Therefore, by the contrapositive of Lemma~\ref{lem:hyp-helly-gap}, there exist $u,v \in \Supp(\pi)$ such that $B_{\lambda_\tau(u)-1-2\delta}(u) \cap B_{\lambda_\tau(v)-1-2\delta}(v) = \emptyset$. 
    In this situation, $d(u,v) \ge \lambda_\tau(u)+\lambda_\tau(v)-4\delta-1$.
    By Claim~\ref{claim:dist-cstar-support}, $C^* \subseteq B_{\lambda_\tau(u)+ 18\delta+3}(u)$. Similarly, $C^* \subseteq B_{\lambda_\tau(v)+ 18\delta+3}(v)$.

    We fix some shortest $(u,v)$-path $P$.
    Then, let $Q$ be the subpath of $P$ that is made of all the vertices $w$ such that $\lambda_\tau(u) - 22\delta-4 \le d(u,w) \le \lambda_\tau(u) + 18\delta+3$.
    To prove the claim, we will prove that $d(z,Q) \le 42\delta+7$ for every $z \in C^*$.
    By doing so, $\diam(C^*) \le \diam(Q) +84\delta+14 \le 124\delta+21$. 
    In what follows, let $z \in C^*$ be arbitrary.
    As $d(u,v) \le d(u,z)+d(z,v)$, there must be a vertex $w_z$ of $P$ such that $d(u,w_z) \le d(u,z)$, $d(v,w_z) \le d(v,z)$.
    Recall that $d(u,z) \le \lambda_\tau(u)+18\delta+3$, $d(v,z) \le \lambda_\tau(v)+18\delta+3$, and $d(u,v) \ge \lambda_\tau(u)+ \lambda_\tau(v)-4\delta-1$.
    Consequently, $d(u,z)+d(z,v) \le d(u,v)+\mu$, where $\mu = 40\delta+7$.
    By Lemma~\ref{lem:hyp-morse}, $d(z,w_z) \le \mu + 2\delta = 42\delta+7$.
    So, it remains to show that $w_z \in Q$.
    As $d(u,w_z) \le d(u,z)$, we get $d(u,w_z) \le \lambda_\tau(u)+18\delta+3$.
    Finally,
    \begin{flalign*}
        d(u,w_z) &= d(u,v) - d(v,w_z) \ge d(u,v) - d(v,z) \\
        &\ge (\lambda_\tau(u)+ \lambda_\tau(v)-4\delta-1) - (\lambda_\tau(v)+18\delta+3) \\
        &\ge \lambda_\tau(u) - 22\delta - 4.
    \end{flalign*}
    Therefore, $w_z \in Q$. $\diamond$

    \begin{myclaim}\label{claim-close-centre}
    If $C^* \ne \emptyset$, and $C^* \subseteq B_{124\delta + 21}(c^*)$ for some $c^* \in C^*$, then also $C_{\pi}(G) \subseteq B_{124\delta+21}(c^*)$.
    \end{myclaim}
    The result follows from Claim~\ref{claim-above-alpha} if $\rad_{\pi}(G) \le \tau$.
    Thus from now on, let us assume $\rad_{\pi}(G) = \gamma > \tau$.
    Note that $\lambda_\gamma(v) \ge \lambda_\tau(v)$ for any $v \in \Supp(\pi)$.
    Assume first $\lambda_\gamma(v) \le 2\delta+1$ for some $v \in \Supp(\pi)$.
    Then, $C_\pi(G) \subseteq B_{2\delta+1}(v) \subseteq B_{22\delta+5}(c^*)$ holds, where the first inclusion follows from the definition of $\lambda_\gamma(v)$ and the second inclusion follows from Claim~\ref{claim:dist-cstar-support}.
    From now on, we assume $\lambda_\gamma(v) > 2\delta+1$ for every $v \in \Supp(\pi)$.
    Furthermore, as $\rad_{\pi}(G) = \gamma$, $\bigcap\{ B_{\lambda_\gamma(v)-1}(v) : v \in \Supp(\pi) \} = \emptyset$.
    Hence, by the contrapositive of Lemma~\ref{lem:hyp-helly-gap}, there exist $u,v \in \Supp(\pi)$ such that $d(u,v) \ge \lambda_\gamma(u)+\lambda_\gamma(v)-4\delta-1$.
    Recall that as $c^* \in C^*$, by Claim~\ref{claim:dist-cstar-support}, $d(u,c^*) \le \lambda_\tau(u)+18\delta+3 \le \lambda_\gamma(u)+18\delta+3$.
    Similarly, $d(v,c^*) \le \lambda_\gamma(v)+18\delta+3$. 

    Let $c \in C_{\pi}(G)$ be arbitrary.
    In particular, $d(u,c) \le \lambda_\gamma(u)$, and in the same way $d(v,c) \le \lambda_\gamma(v)$.
    Since $G$ is $\delta$-hyperbolic, 
    \begin{flalign*}
	   d(u,v) + d(c,c^*) &\le \max\{d(u,c)+d(v,c^*),d(u,c^*)+d(v,c)\} + 2\delta \\
	   &\le \lambda_\gamma(u) + \lambda_\gamma(v) + 20\delta+3 \\
	   &\le d(u,v) + 24\delta+4.
    \end{flalign*}
    As a result, $d(c,c^*) \le 24\delta+4 < 124\delta+21$. $\diamond$

\smallskip
\noindent
We are now ready to prove the correctness of the algorithm.
For any $i$ such that $M_i \ne \emptyset$, as $C_{\pi_i}(G) \subseteq B_{5\delta+1}(c^i)$, we get $\rad_{\pi_i}(G) \ge r_{\pi_i}(c^i)-5\delta-1$.
Thus, if $r_{\pi_i}(c^i) > i +7\delta+1$, then we get $\rad_{\pi_i}(G) > i+2\delta$.
By the contrapositive of Claim~\ref{claim:rad-ci}, $\rad_\pi(G) > \tau$.
Assume now $r_{\pi_i}(c^i) \le i +7\delta+1$ for each $i$.
If $C^* = \emptyset$, then, by the contrapositive of Claim~\ref{claim-above-alpha}, $\rad_\pi(G) > \tau$.
Otherwise, let $c^* \in C^*$ be arbitrary. 
If $C^* \not\subseteq B_{124\delta + 21}(c^*)$, then, we obtain $\diam(C^*) > 124\delta+21$; by the contrapositive of Claim~\ref{claim-below-alpha}, $\rad_{\pi}(G) < \tau$.
However, if $C^* \subseteq B_{124\delta + 21}(c^*)$, then by Claim~\ref{claim-close-centre}, $C_\pi(G) \subseteq B_{124\delta + 21}(c^*)$.
\end{proof}

\section{Exact algorithms}\label{exact}
We present (almost) linear-time algorithms for the \textsc{Weighted Center} problem.
With the notable exception of Sec.~\ref{planar}, all the results in what follows are based on our Theorem~\ref{th:hyp}, used in combination with other properties of the classes of graphs considered.
Furthermore, for the classes of graphs with unbounded hyperbolicity, we present parameterized algorithms, where the parameter is the hyperbolicity value.

\subsection{Bounded-degree graphs.}\label{max-degree}
The {\em growth function} of a graph $G=(V,E)$ is defined as $\gamma_G : \rho \in \mathbb{N} \mapsto \max\{ |B_\rho(v)| : v \in V \}$.
Note that $\gamma_G(\rho) = \cO(\Delta^\rho)$ if $G$ has maximum degree $\Delta$.
A graph $G$ of \emph{polynomial growth} is one such that $\gamma_G(\rho) = \cO(\rho^\alpha)$, for some fixed exponent $\alpha$.
See~\cite{Pap}, and the references therein, for previous studies on the graphs of polynomial growth.
Finally, a  graph $G$ has \emph{doubling dimension} $\beta$ if every ball in $G$ of positive radius is included in the union of at most $2^\beta$ balls with half-radius.
Then, $\gamma_G(\rho) = \cO(\rho^{\beta})$. 

\begin{mytheorem}\label{th:growth}
    For any $\delta$-hyperbolic graph $G$, and for any profile $\pi$, the \textsc{Weighted Center} problem can be solved in $\cO((\gamma_G(\cO(\delta))+\delta\log^2{n})m)$ time, with $\gamma_G$ the growth function.
    In particular, this algorithm runs
    \begin{enumerate}
        \item\label{item-max-degree} in $\cO((\Delta^{O(\delta)}+\delta\log^2{n})m)$ time, if $G$ has maximum degree $\Delta$;
        \item\label{item-poly-growth} in $\cO(\delta^\alpha m \log^2{n})$ time, if $G$ has polynomial growth with exponent $\alpha$;
        \item\label{item-doubling} in $\cO(\delta^\beta m \log^2{n})$ time, if $G$ has doubling dimension $\beta$.
    \end{enumerate}
\end{mytheorem}
\begin{proof}
    We apply Theorem~\ref{th:hyp}(\ref{item:hyp-2}) to compute a vertex $c_2$ such that $C_\pi(G) \subseteq B_{124\delta+21}(c_2)$.
    To compute $C_\pi(G)$, it suffices to compute $r_\pi(x)$, for every $x \in B_{124\delta+21}(c_2)$, which requires $\cO(m)$ time per vertex.
    There are at most $\gamma_G(124\delta+21)$ vertices in the ball.
\end{proof}

\subsection{Chordal graphs, and weakly bridged graphs.}\label{weakly-bridged}
Recall that a \emph{chordal graph} is one with no induced cycle of length more than three.
A graph $G=(V,E)$ is called \emph{weakly bridged} if it has convex balls and if any of its induced cycles of length five is contained in the neighborhood of some vertex.
Note that chordal graphs are weakly bridged.
However, chordal graphs are $1$-hyperbolic~\cite{HypChordal}, whereas the hyperbolicity of weakly bridged graphs is unbounded~\cite{HypBridged}.
The properties of weakly bridged graphs were studied in~\cite{ChOs}.
In particular, they satisfy the following \emph{Triangle Condition}:
\begin{itemize}
    \item $\TC$ If $u \sim v$ and $d(u,w) = d(v,w)$, for some $u,v,w \in V$, then $S_1(u,w) \cap S_1(v,w) \ne \emptyset$.
\end{itemize}
We will need the following characterization of graphs with convex balls:
\begin{mylemma}[see Theorem 2 in~\cite{SoCh1983}]\label{lem:convex}
    All balls $B_k(v)$ $(v\in V$, $k\geq 1)$ of a graph $G$ are convex if and only if $G$ does not contain isometric cycles of length $\ell > 5$, and for any two vertices $x,y$ of $G$ the slice $S_1(x,y)$ is a clique.  
\end{mylemma}
The \textsc{Weighted Center} problem on weakly bridged graphs was studied in~\cite{Gpunimodal-ecc}, Sec. $5$. In particular, every radius function of a weakly bridged graph is $2$-weakly peakless (see~\cite{Gpunimodal-ecc}, Proposition $5.5$).
In what follows, we gather some more results from~\cite{Gpunimodal-ecc}; we provide more precise statements than the authors in~\cite{Gpunimodal-ecc}, based on a careful analysis of their proofs. 
Let $G=(V,E)$ be weakly bridged, and let $\pi$ be an arbitrary profile:
\begin{itemize}
    \item (WB.1) for any $v \in V$ such that $v$ is \emph{not} a local minimum of $r_\pi$ in $G$, we can compute in $\cO(m)$ time the subset of all its neighbors $u \in N(v)$ such that $r_\pi(u) < r_\pi(v)$, and the corresponding values $r_\pi(u)$ (follows from~\cite{Gpunimodal-ecc}, Lemma $5.9$ and its proof);
    \item (WB.2) for any $v \in V$ such that $v$ is a local minimum of $r_\pi$ in $G$, but $v \notin C_\pi(G)$, a vertex $u \in B_2(v)$ such that $r_\pi(u) < r_\pi(v)$ can be computed in $\cO(m)$ time; furthermore, for any $c \in C_\pi(G)$, there exists some $u_c^* \in S_2(v,c)$ such that $u \in N[u_c^*]$ (follows from~\cite{Gpunimodal-ecc}, Proposition $5.11$ and its proof\footnote{
        More specifically, for some $u_c^* \in S_2(v,c)$ the authors in~\cite{Gpunimodal-ecc} proved that $r_\pi(u_c^*) < r_\pi(v)$, and that $v,u,u_c^*$ have some common neighbor $w_{\max}$. 
        As $v$ is a local minimum of $r_\pi$ in $G$, $r_{\pi}(u) < r_\pi(v) \le r_{\pi}(w_{\max})$, and similarly $r_{\pi}(u_c^*) < r_{\pi}(w_{\max})$.
        It implies that for any $x \in F_\pi(w_{\max})$, $u,u_c^* \in S_1(w_{\max},x)$ holds.
        By Lemma~\ref{lem:convex}, either $u=u_c^*$, or $u \sim u_c^*$.
        }).
\end{itemize}
Starting from any vertex $v_0$, we can repeatedly apply either (WB.1) or (WB.2) until we find a center.
By doing so, we construct an improving sequence $v_0,v_1,\ldots,v_\ell$ for which $r_\pi$ is monotonically decreasing, and $v_\ell \in C_\pi(G)$.
Roughly, to bound the length $\ell$ of the sequence, we must enforce that $d(v_{i+1},C_\pi(G)) < d(v_i,C_\pi(G))$ at each step $i$ of the process. 

\begin{mylemma}\label{lem:wb-best-response}
    Let $G=(V,E)$ be weakly bridged, let $\pi$ be an arbitrary profile, and let $v \in V$.
    If $v$ is not a local minimum of $r_\pi$ in $G$, then, for any $c \in C_\pi(G)$, there exists some neighbor $u_c$ of $v$ such that $r_\pi(u_c)$ is minimum within $N[v]$, and $d(u_c,c) < d(v,c)$.
\end{mylemma}
\begin{proof}
    %In what follows, let $c \in C_\pi(G)$ be closest to $v$.
    Let $u_0$ be minimizing $r_\pi$ within $N[v]$.
    As $v$ is a not a local minimum of $r_\pi$ in $G$, $r_\pi(u_0) < r_\pi(v)$.
    Suppose by contradiction $d(v,c) < d(u_0,c)$.
    For every $x \in \Supp(\pi)$, as the balls in $G$ are convex, we get $d(w,x) \le \max\{d(u_0,x),d(c,x)\}$ for every $w \in I(u_0,c)$.
    In particular, $r_\pi(w) \le \max\{r_\pi(u_0),r_\pi(c)\} = r_\pi(u_0) < r_\pi(v)$.
    As $v \in I(u_0,c)$, a contradiction arises.
    Therefore, $d(u_0,c) \le d(v,c)$.
    If $d(u_0,c) < d(v,c)$, then we set $u_c := u_0$, and we are done.
    Otherwise, by $\TC$, there exists some $u_1 \in S_1(u_0,c) \cap S_1(v,c)$.
    By similar convexity arguments as above, $r_\pi(u_1) \le r_\pi(u_0)$.
    Due to the minimality of $r_\pi(u_0)$ within $N[v]$, we obtain $r_\pi(u_1) = r_\pi(u_0)$.
    We set $u_c := u_1$.
\end{proof}

%The next lemma is the cornerstone of our approach in this section:
\begin{mylemma}\label{lem:wb-local-search}
    Let $G=(V,E)$ be weakly bridged, let $\pi$ be an arbitrary profile, and let $v \in V$.
    If $v \notin C_\pi(G)$, then a vertex $u \in B_3(v)$ such that $r_\pi(u) < r_\pi(v)$, \underline{and} $d(u,C_\pi(G)) < d(v,C_\pi(G))$, can be computed in $\cO(m)$ time.
\end{mylemma}
\begin{proof}
    We start with the following claim:
    \begin{myclaim}\label{claim:wb-loca-min-g}
        If (WB.2) applies to $v$, then it returns a vertex $u \in B_2(v)$ such that $r_\pi(u) < r_\pi(v)$, and $d(u,C_\pi(G)) < d(v,C_\pi(G))$.
    \end{myclaim}
        Indeed, let $c \in C_\pi(G)$ be closest to $v$.
        Recall that there exists some $u_c^* \in S_2(v,c)$ such that $u \in N[u_c^*]$.
        Hence, $d(u,C_\pi(G)) \le d(u,c) \le 1+d(u_c^*,c) = d(v,c)-1$. $\diamond$

    \smallskip
    \noindent
    We now describe our algorithm for an arbitrary vertex $v$.
    If $v$ is a local minimum of $r_\pi$ in $G$, then we are done by Claim~\ref{claim:wb-loca-min-g}.
    Otherwise, we apply (WB.1) to $v$.
    By doing so, we compute the subset $X(v)$ of all its neighbors minimizing $r_\pi$.
    As $X(v) \subseteq S_1(v,z)$ for each $z \in F_\pi(v)$, by Lemma~\ref{lem:convex}, $X(v)$ is a clique.
    Furthermore, we observe that $d(u',C_\pi(G)) \le d(v,C_\pi(G))$ for every $u' \in X(v)$.
    Indeed, this is because $X(v)$ is a clique and, by Lemma~\ref{lem:wb-best-response} applied for any $c \in C_\pi(G)$ closest to $v$, some unknown vertex in $X(v)$ is at distance $d(v,C_\pi(G))-1$ to the center.
    We pick some arbitrary $u_0 \in X(v)$.
    If (WB.1), (WB.2) cannot be applied to $u_0$, then necessarily $u_0 \in C_\pi(G)$. In this situation, we return $u := u_0$.
    Assume now that (WB.2) applies to $u_0$. By Claim~\ref{claim:wb-loca-min-g}, the latter returns some $u_1 \in B_2(u_0) \subseteq B_3(v)$ such that: $r_\pi(u_1) < r_\pi(u_0) < r_\pi(v)$; and $d(u_1,C_\pi(G)) < d(u_0,C_\pi(G)) \le d(v,C_\pi(G))$. In this situation, we return $u := u_1$.
    Otherwise, we apply (WB.1) to $u_0$ to compute the subset $X(u_0)$ of all its neighbours minimizing $r_\pi$.
    By similar convexity arguments as above, $X(u_0)$ is a clique.
    Furthermore, $d(u'',C_\pi(G)) \le d(u_0,C_\pi(G)) \le d(v,C_\pi(G))$ for every $u'' \in X(u_0)$.
    Let $u_2$ be maximizing $|N(u_2) \cap X(v)|$ within $X(u_0)$.
    In this situation, we return $u := u_2$.
    Indeed, we first note that $r_\pi(u_2) < r_\pi(u_0) < r_\pi(v)$.
    Suppose now by contradiction $d(u_2,C_\pi(G)) \ge d(v,C_\pi(G))$.
    Then, $d(u_2,C_\pi(G)) = d(u_0,C_\pi(G)) = d(v,C_\pi(G)) = k$.
    Let $c_0 \in C_\pi(G)$ be closest to $v$.
    By Lemma~\ref{lem:wb-best-response} applied to both $v,u_0$, there exist $x_0 \in X(v)$, $x_2 \in X(u_0)$ such that $d(x_0,c_0) = d(x_2,c_0) = k-1$.
    As $x_0,x_2 \in S_1(u_0,c_0)$, by Lemma~\ref{lem:convex}, we get $x_0 \sim x_2$.
    However, as $G$ has convex balls, there is no induced $C_4$.
    It implies that the subsets $N(u'') \cap X(v)$, for $u'' \in X(u_0)$, must be comparable for inclusion.
    Therefore, due to the maximality of $|N(u_2) \cap X(v)|$, $u_2 \sim x_0$ holds.
    It implies that $x_0 \in S_1(u_2,c_0)$.
    But then, for every $y \in \Supp(\pi)$, as $G$ has convex balls, $d(w,y) \le \max\{d(u_2,y),d(c_0,y)\}$ holds for every $w \in I(u_2,c_0)$.
    So, $r_\pi(w) \le \max\{r_\pi(u_2),r_\pi(c_0)\} = r_\pi(u_2)$ for every $w \in I(u_2,c_0)$.
    As $x_0 \in I(u_2,c_0)$, and $r_\pi(x_0) = r_\pi(u_0) > r_\pi(u_2)$, a contradiction arises.
\end{proof}

We can now state the main result of this section.
Theorem~\ref{th:hyp}(\ref{item:hyp-1}) is used for selecting the start vertex of our local search.
To improve the runtime for chordal graphs, we use a stronger version of this result (Theorem~\ref{thm:hyp-weak}), in combination with the following property of this class of graphs:
\begin{mylemma}[see Theorem $1$ in~\cite{BrChDr}]\label{lem:lambda-chordal}
    Every chordal graph $3$-embeds into a tree. 
\end{mylemma}

\begin{mytheorem}\label{thm:weakly-bridged}
    For every weakly bridged graph $G$, for any profile $\pi$, the \textsc{Weighted Center} problem can be solved
    \begin{enumerate}
            \item\label{item-wb-hyp} in $\cO(\delta m\log{n})$ time if $G$ is $\delta$-hyperbolic;
            \item\label{item-chordal} and in $\cO(m)$ time if $G$ is chordal.
        \end{enumerate}
\end{mytheorem}
\begin{proof}
        We compute some vertex $c_1$ such that $d(c_1,C_\pi(G)) \in \cO(\delta\log{n})$, whose existence follows from Theorem~\ref{th:hyp}(\ref{item:hyp-1}).
        If $G$ is chordal, then by the combination of Lemma~\ref{lem:lambda-chordal} with Theorem~\ref{thm:hyp-weak}, the vertex $c_1$ can be chosen such that $d(c_1,C_\pi(G)) \in \cO(1)$.
        Then, starting from $v := c_1$, we repeatedly apply Lemma~\ref{lem:wb-local-search} until we find a center.
        The number of steps of the local search is at most $d(c_1,C_\pi(G))$. Furthermore, each step can be done in $\cO(m)$ time.
\end{proof}

\subsection{(Bipartite) Helly graphs}\label{helly}
We obtain almost optimal algorithms for the \textsc{Weighted Center} problem on dually chordal graphs, and chordal bipartite graphs.
A graph is \emph{dually chordal} if it is the intersection graph of the maximal cliques of some chordal graph.
The same as chordal graphs can be characterized by the existence of a perfect elimination ordering, the dually chordal graphs can be also characterized by the existence of a special type of dismantling ordering~\cite{BrChDrDuallyChordal}.
A \emph{chordal bipartite graph} is a bipartite graph with no induced cycle of length more than four.
Both the dually chordal graphs and the chordal bipartite graphs are $1$-hyperbolic.

We also consider some of their metric generalizations.
Namely, a graph is called \emph{Helly} if every family of pairwise intersecting balls has a nonempty common intersection.
A \emph{half-ball} in a bipartite graph is the intersection of a ball with one of its two partite sets.
A bipartite graph is called \emph{bipartite Helly} if and only if every family of pairwise intersecting half-balls has a nonempty common intersection.
Interestingly, every (bipartite) graph $G$ can be isometrically embedded in a smallest (bipartite) Helly graph $\mathcal{H}(G)$, which is called its \emph{injective hull}.
Furthermore, dually chordal graphs and chordal bipartite graphs are Helly and bipartite Helly, respectively.

\begin{mytheorem}\label{th:helly}
    For a (bipartite) Helly graph, the \textsc{Weighted Center} problem can be solved
    \begin{enumerate}
        \item\label{item-helly-gal} in $\cO(\delta^3m\log^2{n})$ time, if it is $\delta$-hyperbolic;
        \item\label{item-chordal-like} in $\cO(m\log{n})$ time, if it is dually chordal or chordal bipartite.
    \end{enumerate}
\end{mytheorem}

Roughly, we use Theorem~\ref{th:hyp} to restrict the search for a center to some ball of radius $\cO(\delta)$.
Then, for any profile $\pi$, for any value $\tau$, we can identify the vertices $u$ of this ball such that $r_\pi(u) \le \tau$ using a nontrivial combination of prior results for binary profiles on these classes of graphs and a generalization of the techniques from~\cite{HellyHyp}, Section $4.1$.

Our main tool for Helly graphs is the following Lemma~\ref{lem:helly-for-DH}:
    \begin{mylemma}[see Theorem $2$ in~\cite{Helly}]\label{lem:helly-for-DH}
        Let $H = (V,E)$ be a Helly graph.
        For any subset $M$, and for any integer $k$, we can compute $\bigcap\{ B_k(x) :x \in M \}$ in $\cO(km)$ time.
    \end{mylemma}

    There is an analogous result for bipartite Helly graphs:
    \begin{mylemma}[see Lemma $4.5$ in~\cite{Gpunimodal-ecc}]\label{lem:bip-Helly}
        Let $G = (V_0 \cup V_1,E)$ be a bipartite Helly graph.
        For any subset $M$, and for any integer $k$, we can compute $\bigcap\{ B_k(x) :x \in M \}$ in $\cO(km)$ time.
    \end{mylemma}

    Finally, for dually chordal graphs and chordal bipartite graphs, we also use the following results on tree embeddings to improve our runtime.
    \begin{mylemma}\label{lem:lambda}
		The least value $\lambda$ such that a graph $G=(V,E)$ $\lambda$-embeds into a tree satisfies:
		\begin{enumerate}
			\item\label{item:lambda-dually-chordal} $\lambda \le 4$ if $G$ is dually chordal (see Theorem $3$ in~\cite{BrChDr});
			\item\label{item:lambda-chordal-bips} $\lambda \le 5$ if $G$ is chordal bipartite (see~\cite{ChDr}).
		\end{enumerate}
	\end{mylemma}

    \begin{proof}[Proof of Theorem~\ref{th:helly}]
        In what follows, let $\pi$ be an arbitrary profile.
        We compute a vertex $c^*$, and a radius $\rho$, such that $C_\pi(G) \subseteq B_\rho(c^*)$.
        For that, if $G$ is either dually chordal, or chordal bipartite, then we apply Theorem~\ref{thm:hyp-weak}, in combination with Lemma~\ref{lem:lambda}.
        It takes $\cO(m)$ time, and ensures that $\rho = \cO(1)$.
        Otherwise, we apply Theorem~\ref{th:hyp}(\ref{item:hyp-2}).
        It takes $\cO(\delta m \log^2{n})$ time, and ensures that $\rho = \cO(\delta)$.
        We present an $\cO(\rho^3 m)$-time algorithm for computing all the vertices $v \in B_\rho(c^*)$ such that $r_\pi(v) \le \tau$, for any value $\tau$.
        By Lemma~\ref{lem:weighted-center-decision}, the latter implies an $\cO(\rho^3 m \log{n})$-time algorithm for the \textsc{Weighted Center} problem.

        For that, we reuse some of the terminology of Theorem~\ref{thm:hyp-strong}.
        More precisely, for $v \in \Supp(\pi)$, let $\ell(v) = d(v,B_\rho(c^*))$: if $d(c^*,v) \le \rho$, then $\ell(v) = 0$, otherwise, $\ell(v) = d(v,c^*)-\rho$.
        Let also $\lambda_\tau(v) = \left\lfloor \frac \tau {\pi(v)} \right\rfloor$: the maximum possible distance between $v$ and any vertex $u$ such that $r_\pi(u) \le \tau$.
        If $\lambda_\tau(v) < \ell(v)$ for some $v \in \Supp(\pi)$, then there is no vertex of $B_\rho(c^*)$ of radius at most $\tau$; we return the empty set.
        If now $\lambda_\tau(v) = 0$ for some $v \in \Supp(\pi) \cap B_\rho(c^*)$, then we either return $\{v\}$ (if $r_\pi(v) \le \tau$), or the empty set.
        Finally, if $\lambda_\tau(v) \ge \ell(v)+2\rho$ for some $v \in \Supp(\pi)$, then, as $B_\rho(c^*) \subseteq B_{\lambda_\tau(v)}(v)$, we can ignore $v$ for the remainder of the algorithm.
        In particular, if $\lambda_\tau(v') \ge \ell(v')+2\rho$ for every $v' \in \Supp(\pi)$, then, we return $B_\rho(c^*)$.
        Thus, from now on, we assume the following for every $v \in \Supp(\pi)$: there exists some $i$ with $0 \le i \le 2\rho-1$ such that $\lambda_\tau(v) = \ell(v)+i$.

        \medskip
        {\bf The main procedure.}
        Fix some $i$ with $0 \le i \le 2\rho-1$, and let $V_i := \{ v \in \Supp(\pi) : \lambda_\tau(v) = \ell(v)+i \}$.
        Next, we describe an $\cO(\rho^2 m)$-time algorithm for computing $X_i := \bigcap\{ B_{\lambda_\tau(v)}(v) \cap B_\rho(c^*) : v \in V_i \}$.
        If we apply this algorithm for each possible value $i$, then we get an $\cO(\rho^3 m)$-time algorithm for returning all the vertices in $B_\rho(c^*)$ of radius at most $\tau$.
        For that, we adapt the ball-growing strategy in~\cite{HellyHyp}, Section $4.1$.
        \begin{itemize}
            \item \emph{In case $G$ is bipartite, we make the additional assumption for what follows that all the vertices of $V_i$ are included in the same partite set.}
        This assumption is only used in the proof of Claim~\ref{claim:helly-4}.
        \end{itemize}
        If $V_i$ intersects both partite sets $U_0,U_1$ of $G$, then we run the same algorithm (presented next) for $V_i \cap U_0$, $V_i \cap U_1$ separately.
        
        Let $k_i = \max\{\ell(v) : v \in V_i\}$. Furthermore, for each $j$ such that $0 \le j \le k_i$, let $V_{i,j} = \{v \in V_i : \ell(v) = j\}$.
        For each $j$ with $0 \le j \le \rho-1$, we directly compute $X_{i,j} = \bigcap\{ B_{j+i}(v)  : v \in V_{i,j} \}$.
        This can be done in $\cO(jm) = \cO(\rho m)$ time by either applying Lemma~\ref{lem:helly-for-DH}, if $G$ is Helly, or Lemma~\ref{lem:bip-Helly}, if $G$ is bipartite Helly.
        In particular, $X_i^* = \bigcap\{ X_{i,j} \cap B_\rho(c^*) : 0 \le j \le \rho-1 \}$ can be computed in $\cO(\rho^2 m)$ time.
        If $k_i < \rho$, then $X_i = X_i^*$, and so we are done.
        Otherwise, we get $X_i = \bigcap\{ B_{\lambda_\tau(v)}(v) \cap X_i^* : v \in V_i, \ell(v) \ge \rho \}$.
        Hence, we may assume for the remainder of the algorithm $\ell(v) \ge \rho$ for every $v \in V_i$.

        \smallskip
        \noindent
        %Based on Claim~\ref{claim:helly-1}, 
        We present a two-stage procedure for dealing with the vertices $v \in V_i$ such that $\ell(v) \ge \rho$.
        For convenience, in what follows let $V_{i,\ge j} = \bigcup_{j'=j}^{k_i}V_{i,j'}$.
        Furthermore, let $L_t = \{ z \in V : d(c^*,z) = t\}$.
        By a partial partition of $L_t$, we mean a partition of some nonempty subset of $L_t$.

       \smallskip
        \noindent
       \underline{First stage}. For each $j$ with $0 \le j \le k_i-\rho$, we compute a partial partition $Z_1^j,Z_2^j,\ldots,Z_{q_j}^j$ of $L_{k_i-j+\rho}$ such that, for some corresponding partition $Y_1^j,Y_2^j,\ldots,Y_{q_j}^j$ of $V_{i,\ge k_i-j}$, the following property holds for any $1 \le t \le q_j$: $$Z_t^j = \bigcap\{ S_{k_i-j+\rho}(c^*,v) : v \in Y_t^j \}.$$
        For $j = 0$, let $V_{i,k_i} = \{v_1^0,v_2^0,\ldots,v_{q_0}^0\}$ be totally ordered.
        We define $q_0$ groups such that, for any $0 \le t \le q_0$, $Z_t^0 = \{v_t^0\}$.
        For the corresponding partition of $V_{i,k_i}$, we also have $Y_t^0 = \{v_t^0\}$.

        Assume $j > 0$.
        We compute the new partitions at step $j$ from the partitions computed at step $j-1$.
        More precisely, for any $t$ with $0 \le t \le q_{j-1}$, we set $W_t^{j-1} = N(Z_t^{j-1}) \cap L_{k_i-j+\rho}$.
        Furthermore, let $V_{i,k_i-j} = \{v_1^j,v_2^j,\ldots,v_{p_j}^j\}$ be totally ordered. 
        For any $t'$ with $1 \le t' \le p_j$, we set $Y_{q_{j-1}+t'}^{j-1} = W_{q_{j-1}+t'}^{j-1} = \{v_{t'}^j\}$.
        By doing so, $Y_1^{j-1},Y_2^{j-1},\ldots,Y_{q_{j-1}+p_j}^{j-1}$ becomes a partition of $V_{i,\ge k_i-j}$.
        However, the sets $W_1^{j-1},W_2^{j-1},\ldots,W_{q_{j-1}+p_j}^{j-1}$ are not necessarily disjoint.
        So, we set $$s := 1, \ \ \mathcal{Y}^{j-1} :=\{1,2,\ldots,q_{j-1}+p_j\}, \ \ \mathcal{W}^{j-1} := \bigcup\{ W_t^{j-1} : 1 \le t \le q_{j-1}+p_j \},$$ then we proceed as follows:
        \begin{itemize}
            \item While $\mathcal{Y}^{j-1} \ne \emptyset$, we pick some $u_s \in \mathcal{W}^{j-1}$ that is in a maximum number of groups $W_t^{j-1}$, for $t \in \mathcal{Y}^{j-1}$.
            We create a new group $Z_s^j = \bigcap\{ W_t^{j-1} : u_s \in W_t^{j-1} \}$. 
            Then, $Y_s^j = \bigcup\{ Y_t^{j-1} : u_s \in W_t^{j-1} \}$.
            For each $t$ such that $u_s \in W_t^{j-1}$, we remove $t$ from $\mathcal{Y}^{j-1}$, and we remove every vertex of $W_t^{j-1}$ from $\mathcal{W}^{j-1}$.
            Finally, we set $s := s+1$.
        \end{itemize}
        Correctness follows from Claims~\ref{claim:helly-2} and~\ref{claim:helly-3} below.
        
        \smallskip
        \noindent
        \underline{Second stage}.  For each $j$ with $k_i-\rho+1 \le j \le k_i+i$, we compute a partial partition $Z_1^j,Z_2^j,\ldots,Z_{q_j}^j$ of $V$ such that, for some partition $Y_1^j,Y_2^j,\ldots,Y_{q_j}^j$ of $V_{i,\ge \rho}$, the following property holds for any $1 \le t \le q_j$:
        \begin{itemize}
            \item ($\alpha$)  $Z_t^j = \bigcap\{ B_{\ell(v)-k_i+j}(v) \cap B_{3\rho+j-k_i}(c^*) : v \in Y_t^j \}$.
        \end{itemize}
        We construct the partition at step $j$ from that at step $j-1$.
        In particular, for $j=k_i-\rho+1$, we start from the last partition obtained during the first phase.
        More precisely, for any $t$ with $0 \le t \le q_{j-1}$, we set $W_t^{j-1} = Z_t^{j-1} \cup N(Z_t^{j-1})$.
        Then (similarly to the first phase) we set $$s := 1, \ \ \mathcal{Y}^{j-1} :=\{1,2,\ldots,q_{j-1}\}, \ \ \mathcal{W}^{j-1} := \bigcup\{ W_t^{j-1} : 1 \le t \le q_{j-1} \},$$ and we proceed as follows:
        \begin{itemize}
            \item While $\mathcal{Y}^{j-1} \ne \emptyset$, we pick some $u_s \in \mathcal{W}^{j-1}$ that is in a maximum number of groups $W_t^{j-1}$, for $t \in \mathcal{Y}^{j-1}$.
            We create a new group $Z_s^j = \bigcap\{ W_t^{j-1} : u_s \in W_t^{j-1} \}$. 
            Then, $Y_s^j = \bigcup\{ Y_t^{j-1} : u_s \in W_t^{j-1} \}$.
            For each $t$ such that $u_s \in W_t^{j-1}$, we remove $t$ from $\mathcal{Y}^{j-1}$, and we remove every vertex of $W_t^{j-1}$ from $\mathcal{W}^{j-1}$.
            Finally, we set $s := s+1$.
        \end{itemize}
        Correctness follows from Claim~\ref{claim:helly-4} below.

        \smallskip
        \noindent
        When the second phase ends, if $q_{k_i+i} > 1$, then $X_i = \emptyset$; otherwise, $X_i = Z_1^{k_i+i} \cap B_\rho(c^*)$.
        Correctness follows from Claim~\ref{claim:helly-5} below.

        \medskip
        \noindent
        {\bf Implementation.}
        Note that we do not need to construct the partitions $Y_1^j,\ldots,Y_{q_j}^j$ explicitly.
        Hence, we only need to detail the computation of $Z_1^j,\ldots,Z_{q_j}^j$ at every step $j$.
        For $j = 0$, this is in $\cO(n)$.
        Assume now $j > 0$.
        To compute the new groups $Y_s^j$ from the $W_t^{j-1}$'s, we maintain the vertices of $\mathcal{W}^{j-1}$ in a max-heap, where the priority of a vertex $u$ is equal to the current number of groups $W_t^{j-1}$, for $t \in \mathcal{Y}^{j-1}$, in which it appears. Since all the priorities are integers, and the maximum priority is non-increasing, we can use for our implementation of max-heap a bucket queue~\cite{BucketQueue}. By doing so, the running time for Step $j$ is linear in $\sum_t |W_t^{j-1}|$ (we refer to Theorem $2$ in~\cite{Helly}, and to Lemma $6.5$ in~\cite{Gpunimodal-ecc}, for similar discussions).
        As this sum is always in $\cO(n+m)$, we can upper bound the total running time for the second stage of the algorithm by an $\cO(\rho m)$.
        However, for any step $j > 0$ during the first stage, we get a more precise upper bound in $\cO(m_{j-1,j}+n_j)$, where $n_j = |L_{k_i-j+\rho}|$ and $m_{j-1,j}$ is the number of edges between $L_{k_i-(j-1)+\rho}$ and $L_{k_i-j+\rho}$.
        In particular, the total running time for the first stage of the algorithm is linear in $n + \sum_{j=1}^{k_i-\rho}(m_{j-1,j}+n_j) = \cO(m+n)$.

        \medskip
        \noindent
        {\bf Correctness.} We prove by induction that all the required properties for $Z_1^j,\ldots,Z_{q_j}^j$ hold at every step $j$.
        This is straightforward if $j = 0$.
        Furthermore, for $j > 0$ during the first stage, we prove that
        \begin{myclaim}\label{claim:helly-2}
            For $t$ with $0 \le t \le q_{j-1}+p_j$, $W_t^{j-1} = \bigcap\{ S_{k_i-j+\rho}(c^*,v) : v \in Y_t^{j-1} \}$.
        \end{myclaim}
            Again, this is straightforward for $t > q_{j-1}$ (singleton groups of $V_{i,k_i-j}$).
            Thus, assume $t \le q_{j-1}$.
            By the induction hypothesis, $Z_t^{j-1} = \bigcap\{ S_{k_i-(j-1)+\rho}(c^*,v) : v \in Y_t^{j-1} \}$.
            As $W_t^{j-1} = N(Z_t^{j-1}) \cap L_{k_i-j+\rho}$, we get $W_t^{j-1} \subseteq \bigcap\{ S_{k_i-j+\rho}(c^*,v) : v \in Y_t^{j-1} \}$.
            Conversely, let $z \in \bigcap\{ S_{k_i-j+\rho}(c^*,v) : v \in Y_t^{j-1} \}$ be arbitrary.
            As also $Z_t^{j-1} \ne \emptyset$, the balls in $\{B_{k_i-(j-1)+\rho}(c^*), B_1(z)\} \cup \{ B_{\ell(v)-k_i+(j-1)}(v) : v \in Y_t^{j-1} \}$ pairwise intersect.
            Furthermore, if $G$ is bipartite, then let $U_b$ be its partite set such that $L_{k_i-(j-1)+\rho} \subseteq U_b$.
            Again, as $z \in \bigcap\{ S_{k_i-j+\rho}(c^*,v) : v \in Y_t^{j-1} \}$ and $Z_t^{j-1} \ne \emptyset$, the half-balls in $\{B_{k_i-(j-1)+\rho}(c^*) \cap U_b, B_1(z) \cap U_b\} \cup \{ B_{\ell(v)-k_i+(j-1)}(v) \cap U_b : v \in Y_t^{j-1} \}$ pairwise intersect.
            Therefore, by the Helly property (applied for balls if $G$ is Helly, and for half-balls if $G$ is bipartite Helly), $z \in N(Z_t^{j-1})$.
            As $z \in L_{k_i-j+\rho}$, we get $z \in W_t^{j-1}$. $\diamond$

        \medskip
        \noindent
        Recall that every group $Z_s^j$ is the nonempty intersection of some groups $W_t^{j-1}$, and that correspondingly $Y_s^j$ is the union of all the associated groups $Y_t^{j-1}$.
        Therefore, Claim~\ref{claim:helly-2} implies that $Z_s^j = \bigcap\{ S_{k_i-j+\rho}(c^*,v) : v \in Y_s^j \}$.
        To complete the proof of correctness of the first stage of the algorithm, it suffices to prove that all the groups $Z_1^j,\ldots,Z_{q_j}^j$ are pairwise disjoint.
        The latter directly follows from the following observation:
        \begin{myclaim}\label{claim:helly-3}
            For every step $j > 0$, after we created a new group $Z_s^j$ and we updated $\mathcal{Y}^{j-1},\mathcal{W}^{j-1}$, then for any remaining $t$, $W_t^{j-1} \cap Z_s^j = \emptyset$ holds.
        \end{myclaim}
            Suppose by contradiction there is a $t$ such that $W_t^{j-1} \cap Z_s^j \ne \emptyset$.
            Let $u_s'$ be any vertex in this intersection.
            Then, $u_s'$ occurs in $W_t^{j-1}$ and in every former group $W_{t_s}^{j-1}$ such that $u_s \in W_{t_s}^{j-1}$.
            But as $u_s \notin W_t^{j-1}$, the latter contradicts the maximality of the number of occurences of $u_s$. $\diamond$
            
         \medskip
         \noindent
        The latter concludes the proof of correctness for the first stage of our algorithm.
        Furthermore, we can also use Claim~\ref{claim:helly-3} for proving that the sets $Z_1^j,\ldots,Z_{q_j}^j$ are pairwise disjoint at every step $j$ during the second stage.
        We now prove that
        \begin{myclaim}\label{claim:helly-4}
           Property ($\alpha$) holds after every step $j$, with $k_i-\rho \le j \le k_i+i$.
        \end{myclaim}
            %This part is similar (but more involved) to what we did for Claim~\ref{claim:helly-1}.        
        We prove the property by induction on $j$.
        Since the first stage of the algorithm was correct, the claim is true if $j = k_i-\rho$.
        Thus, we assume $j > k_i-\rho$.
        By the induction hypothesis, for any $t$ with $0 \le t \le q_{j-1}$, $$Z_t^{j-1} = \bigcap\{ B_{\ell(v)-k_i+j-1}(v) \cap B_{3\rho+j-1-k_i}(c^*) : v \in Y_t^{j-1} \}.$$
        As $W_t^{j-1} = Z_t^{j-1} \cup N(Z_t^{j-1})$, we get $$W_t^{j-1} \subseteq \bigcap\{ B_{\ell(v)-k_i+j}(v) \cap B_{3\rho+j-k_i}(c^*) : v \in Y_t^{j-1} \}.$$
        Conversely, let $w \in \bigcap\{ B_{\ell(v)-k_i+j}(v) \cap B_{3\rho+j-k_i}(c^*) : v \in Y_t^{j-1} \}$ be arbitrary.
        As $Z_t^{j-1} \subseteq W_t^{j-1}$, we may assume $w \notin Z_t^{j-1}$.
        As also $Z_t^{j-1} \ne \emptyset$, the balls in $\{B_{3\rho+j-1-k_i}(c^*), B_1(w)\} \cup \{ B_{\ell(v)-k_i+(j-1)}(v) : v \in Y_t^{j-1} \}$ pairwise intersect.
        Furthermore, if $G$ is bipartite, then let $U_b$ be the partite set that does not contain $w$.
        We prove in the following that the half-balls in $\{B_{3\rho+j-1-k_i}(c^*) \cap U_b, B_1(w) \cap U_b\} \cup \{ B_{\ell(v)-k_i+(j-1)}(v) \cap U_b : v \in Y_t^{j-1} \}$ pairwise intersect:
        \begin{itemize}
            \item If $c^* = w$, then we get $\emptyset \ne N(c^*) = B_1(c^*) \cap B_1(w) \cap U_b$.
            Furthermore, $3\rho+j-1-k_i \ge 2\rho > 1$.
            Otherwise, $\emptyset \ne S_1(w,c^*) \subseteq B_{3\rho+j-1-k_i}(c^*) \cap B_1(w) \cap U_b$.
            In both cases, $B_{3\rho+j-1-k_i}(c^*) \cap U_b$ and $B_1(w) \cap U_b$ intersect.
            \item Now, let $v \in Y_t^{j-1}$ be arbitrary.
            If either $v \ne w$, or $\ell(v)-k_i+j-1 > 0$, then we obtain by similar arguments as above that $B_{\ell(v)-k_i+j-1}(v) \cap U_b$ and $B_1(w) \cap U_b$ intersect.
            However, suppose by contradiction both $v = w$, and $\ell(v)-k_i+j-1 = 0$ hold.
            As $w \notin Z_t^{j-1}$, either $d(v,c^*) = d(w,c^*) = \rho+j-k_i$ or $d(v,v') = d(w,v') = \ell(v')-k_i+j$ for some other $v' \in Y_t^{j-1}$.
            In the former case $B_{\ell(v)-k_i+j-1}(v) \cap B_{\rho+j-1-k_i}(c^*) = \emptyset$, while in the latter case $B_{\ell(v)-k_i+j-1}(v) \cap B_{\ell(v')-k_i+j-1}(v') = \emptyset$.
            Consequently, both cases contradict our assumption that $Z_t^{j-1} \ne \emptyset$.
            \item Again, let $v \in Y_t^{j-1}$ be arbitrary. 
            We want to prove that $B_{\ell(v)-k_i+j-1}(v) \cap U_b$ and $B_{3\rho+j-1-k_i}(c^*) \cap U_b$ intersect. 
            Assume $S_{2\rho}(c^*,v) \subseteq U_b$. Then, we are done as $\emptyset \ne S_{2\rho}(c^*,v) \subseteq B_{\ell(v)-k_i+j-1}(v) \cap B_{3\rho+j-1-k_i}(c^*)$.
            Furthermore, if $S_{2\rho}(c^*,v) \cap U_b = \emptyset$, but $j > k_i-\rho+1$, then we are also done (just take any neighbor of a vertex in the slice $S_{2\rho}(c^*,v)$).
            So, suppose by contradiction $S_{2\rho}(c^*,v) \cap U_b = \emptyset$, and $j = k_i-\rho+1$.
            As every vertex of $S_{2\rho}(c^*,v)$ must be in the same partite set as $u$, $d(v,u) = \ell(v)-\rho+1 = d(v,S_{2\rho}(c^*,v))+1$ is impossible.
            Similarly, as $u$ and $c^*$ must be also in the same partite set, $d(c^*,u) = 2\rho+1$ is also impossible.
            Consequently, $u \in S_{2\rho}(c^*,v)$ holds.
            However, as $w \notin Z_t^{j-1}$, we obtain $d(w,v') = \ell(v')-\rho+1$ for some other $v' \in Y_t^{j-1}$.
            Then, let $w' \in Z_t^{j-1}$ be arbitrary. Observe that $w' \in S_{2\rho}(c^*,v)$ holds.
            Furthermore, as $d(v',w),d(v',w')$ have the same parity, $d(v',w') \le \ell(v')-\rho-1$.
            As it implies $d(v',c^*) \le d(v',w')+d(w',c^*) \le \ell(v')+\rho-1 = d(v',c^*)-1$, a contradiction arises.
            \item Finally, let $v,v' \in Y_t^{j-1}$ be distinct.
            Without loss of generality, $\ell(v) \le \ell(v')$.
            
            Assume $\ell(v)-k_i+j-1 = 0$.
            Since $j-1 \ge k_i-\rho$, and $\ell(v) \ge \rho$, necessarily $\ell(v) = \rho$ and $j = k_i-\rho+1$.
            Suppose by contradiction $w = v$.
            As $w \notin Z_t^{j-1}$, either $d(v,c^*) = d(w,c^*) = \rho+j-k_i$ or $d(v,y) = d(w,y) = \ell(y)-k_i+j$ for some other $y \in Y_t^{j-1}$.
            In the former case $B_{\ell(v)-k_i+j-1}(v) \cap B_{\rho+j-1-k_i}(c^*) = \emptyset$, while in the latter case $B_{\ell(v)-k_i+j-1}(v) \cap B_{\ell(y)-k_i+j-1}(y) = \emptyset$.
            Consequently, both cases contradict our assumption that $Z_t^{j-1} \ne \emptyset$.
            We so conclude that $w \ne v$.
            As $w \in B_{\ell(v)-k_i+j}(v) = B_1(v)$, we get $w \sim v$.
            Therefore, $v \in U_b$.
            It implies that $B_{\ell(v)-k_i+j-1}(v) \cap U_b$ and $B_{\ell(v')-k_i+j-1}(v') \cap U_b$ intersect. 

            From now on, $\ell(v')-k_i+j-1 \ge \ell(v)-k_i+j-1 > 0$.
            Assume further $d(v,v') \le \ell(v')-k_i+j-2$.
            Either $v' \in U_b$, or $N(v') \subseteq U_b$.
            As $N(v') \subseteq B_{\ell(v)-k_i+j-1}(v)$, we get that $B_{\ell(v)-k_i+j-1}(v) \cap U_b$ and $B_{\ell(v')-k_i+j-1}(v') \cap U_b$ intersect. 
            Thus from now on, $d(v,v') \ge \ell(v')-k_i+j-1$.
            As $Z_t^{j-1} \ne \emptyset$, we obtain $d(v,v') \le \ell(v)+\ell(v') + 2(j-1-k_i)$.
            
            Assume first this above inequality to be strict.
            As we assume $\ell(v') \ge \ell(v)$, we obtain $\ell(v')-k_i+j-2 \ge 0$.
            Let $u \in S_{\ell(v')-k_i+j-2}(v',v)$ be arbitrary.
            Note that $d(u,v) \le \ell(v)-k_i+j-1$.
            If $u \in U_b$, then $B_{\ell(v)-k_i+j-1}(v) \cap U_b$ and $B_{\ell(v')-k_i+j-1}(v') \cap U_b$ intersect.
            Furthermore, if $u \notin U_b$, then as $S_1(u,v) \subseteq U_b$, we also get that $B_{\ell(v)-k_i+j-1}(v) \cap U_b$ and $B_{\ell(v')-k_i+j-1}(v') \cap U_b$ intersect.

            Otherwise, $d(v,v') = \ell(v)+\ell(v') + 2(j-1-k_i)$.
            If $S_{\ell(v)-k_i+j-1}(v,v') \subseteq U_b$ then again we get that $B_{\ell(v)-k_i+j-1}(v) \cap U_b$ and $B_{\ell(v')-k_i+j-1}(v') \cap U_b$ intersect.
            Finally, suppose, for the sake of contradiction, $S_{\ell(v)-k_i+j-1}(v,v') \cap U_b = \emptyset$.
            As the vertices of $S_{\ell(v)-k_i+j-1}(v,v')$ are in the same partite set as $w$, $d(v,w)=\ell(v)-k_i+j$ and $d(v',w) = \ell(v')-k_i+j$ are impossible.
            Hence, $w \in S_{\ell(v)-k_i+j-1}(v,v')$.
            Suppose first there exists some $y \in Y_t^{j-1}$ such that $d(y,w) = \ell(y)-k_i+j$.
            {\em Recall that we assume that all the vertices of $S_i$ are in the same partite set}.
            In particular, $\ell(v),\ell(v')$ and $\ell(y)$ have the same parity.
            However, it implies that $d(v,w)$ and $d(y,w)$ have different parities, thereby contradicting that $y,v$ are in the same partite sets.
            Therefore, $d(y,w) \le \ell(y)-k_i+j-1$ holds for any $y \in Y_t^{j-1}$.
            As we assume $w \notin Z_t^{j-1}$, we so deduce that $d(c^*,w) = \rho+j-k_i$.
            However, as $d(c^*,v) = \ell(v)+\rho$ and $d(v,w) = \ell(v)-k_i+j-1$, we obtain that $d(c^*,w)$ and $\rho+j-k_i-1$ must have the same parity.
            A contradiction.
        \end{itemize}
        Overall, by the Helly property (applied for balls if $G$ is Helly, and for half-balls if $G$ is bipartite Helly), $w \in N(Z_t^{j-1})$.
        Therefore, we proved that $$W_t^{j-1} = \bigcap\{ B_{\ell(v)-k_i+j}(v) \cap B_{3\rho+j-k_i}(c^*) : v \in Y_t^{j-1} \}.$$
        Finally, for any $s$ with $1 \le s \le q_j$, let $\mathcal{T}_s \subseteq \{1,2,\ldots,q_{j-1}\}$ be such that $Z_s^j = \bigcap\{ W_t^{j-1} : t \in \mathcal{T}_s \}$, and $Y_s^j = \bigcup\{ W_t^{j-1} : t \in \mathcal{T}_s \}$. 
        We obtain
        \begin{flalign*}
            Z_s^j &= \bigcap_{t \in \mathcal{T}_s} W_t^{j-1} = \bigcap_{t \in \mathcal{T}_s}\bigcap\{ B_{\ell(v)-k_i+j}(v) \cap B_{3\rho+j-k_i}(c^*) : v \in Y_t^{j-1} \} \\
            &= \bigcap\{ B_{\ell(v)-k_i+j}(v) \cap B_{3\rho+j-k_i}(c^*) : v \in \bigcup_{t \in \mathcal{T}_s}Y_t^{j-1} \} = \bigcap\{ B_{\ell(v)-k_i+j}(v) \cap B_{3\rho+j-k_i}(c^*) : v \in Y_s^j \}.
        \end{flalign*}
        As a result, the property ($\alpha$) holds for $j$. $\diamond$

        \medskip
        \noindent
        Finally, we use Property ($\alpha$) to prove that the algorithm correctly computes $X_i$:
        \begin{myclaim}\label{claim:helly-5}
            If $q_{k_i+i} > 1$, then $X_i = \emptyset$.
            Otherwise, $X_i = Z_1^{k_i+i}\cap B_\rho(c^*)$.
        \end{myclaim}
            For any $t$ with $1 \le t \le q_{k_i+i}$, $X_i \subseteq \bigcap\{ B_{\ell(v)+i}(v) \cap B_\rho(c^*) : v \in Y_t^{k_i+i} \}$ holds. By Property ($\alpha$), the latter implies $X_i \subseteq Z_t^{k_i+i}$. Furthermore, as the groups $Z_1^{k_i+i},\ldots,Z_{q_{k_i+i}}^{k_i+i}$ are pairwise disjoint, we either get $X_i = \emptyset$ or $q_{k_i+i} = 1$. Assume $q_{k_i+i} = 1$ for the remainder of the proof. Due to Property ($\alpha$), every vertex of $Z_1^{k_i+i}$ is included in $\bigcap\{ B_{\ell(v)+i}(v) : v \in S_{i,\ge \rho} \}$.
            %Furthermore, as every vertex of $Z_1^{k_i+i}$ is also included in $B_{3\rho+i}(c^*)$, $X_i \subseteq Z_1^{k_i+i}$ holds.
            As we here assume $S_i = S_{i,\ge \rho}$, we get that $X_i = Z_1^{k_i+i} \cap B_\rho(c^*)$. $\diamond$
    \end{proof}

\subsection{Distance-hereditary graphs}\label{distance-hereditary}
A graph $G$ is called \emph{distance-hereditary} if every induced path is a shortest-path.
Note that the distance-hereditary graphs are $1$-hyperbolic~\cite{WuZh01}.
\begin{mytheorem}\label{thm:DH}
    The \textsc{Weighted Center} problem can be solved in $\cO(m)$ time for distance-hereditary graphs.        
\end{mytheorem}

We prove new metric properties of distance-hereditary graphs, which we use in the design of our $\cO(m)$-time algorithm.
In particular, radius functions are $2$-weakly peakless.
Therefore, a local-search algorithm for computing a center can be derived.
To implement each step of the local search in $\cO(m)$ time, we use the injective hull.
In fact, we need the result from~\cite{InjHullGraphs} that the injective hull of a distance-hereditary graph is both Helly and distance-hereditary.
Furthermore,~\cite{InjHullGraphs} also provides an $\cO(m)$-time algorithm for computing the injective hull of a distance-hereditary graph.

    \begin{mylemma}\label{lem:DH-outergates}
        For a distance-hereditary graph $G$, let $H$ be one of its connected induced subgraphs.
        For any $v \notin V(H)$, for any $x_v \in V(H)$ closest to $v$, $S_1(x_v,v) \subseteq \bigcap\{ I(v,x) : x \in V(H) \}$.
    \end{mylemma}
    \begin{proof}
        Consider some arbitrary $v^* \in S_1(x_v,v)$, and $x \in V(H)$.
        As $G$ is distance-hereditary, to prove that $v^* \in I(v,x)$, it suffices to prove that there is an induced $(v,x)$-path containing $v^*$.
        For that, consider the union $P = P_1 \cup P_2$ of a shortest $(v,v^*)$-path $P_1$ in $G$ with a shortest $(x_v,x)$-path in $H$.
        As we assume $x_v$ is a closest-to-$v$ vertex in $V(H)$, the only vertex $v'$ of $P_1\setminus v$ such that $N[v'] \cap V(H) \ne \emptyset$ is $v'=v^*$.
        Therefore, any chord of $P$ must be an edge incident to $v^*$ and to some $x' \in V(P_2)$.
        By choosing $x'$ such that $d(x',x)$ is minimized, we can extract from $P$ an induced $(v,x)$-path $P^*$ containing $v^*$.
    \end{proof}

    We now prove the following property of radius functions in this class of graphs:
    \begin{myproposition}\label{prop:DH}
        For a distance-hereditary graph $G$, and a profile $\pi$, let $u$ and $v$ be two vertices such that $r_\pi(u) \ge r_\pi(v)$. 
        If $d(u,v) \ge 3$, then $r_\pi(w) \le r_\pi(u)$ holds for any $w \in S_1(u,v)$, and equality holds only if $r_\pi(u) = r_\pi(w) = r_\pi(v)$.
        In particular, $r_\pi$ is $2$-weakly peakless.
    \end{myproposition}
    \begin{proof}
        Let $x \in \Supp(\pi)$ be arbitrary.
        To prove the lemma, it suffices to prove that either $d(x,w) < d(x,u)$, or $d(x,w) \le d(x,v)$.
        For that, as $w \in S_1(u,v)$, we get that $uw$ is the first edge of some shortest $(u,v)$-path $P$.
        If $x \in V(P)$, then $d(x,w) < d(x,u)$, unless if $x=u$. Furthermore, if $x=u$, then $d(x,w) < d(x,v)$.
        Therefore, from now on we assume $x \notin V(P)$.
        By Lemma~\ref{lem:DH-outergates}, applied for $H=P$, there exists some $x^* \in \bigcap\{ I(x,y) : y \in V(P) \}$ that is adjacent to at least one vertex on $P$.
        If $x^* \nsim u$, then $w \in S_1(u,x^*)$, and so, $d(x,w) < d(x,u)$.
        Assume now $x^* \sim u$.
        As $P$ is a shortest path, the only possible neighbors of $x^*$ on $P$ are $u$, $w$, and the second neighbor of $w$ on $P$.
        In particular, $d(v,x^*) \ge d(v,w) = d(v,u)-1 \ge 2 \ge d(w,x^*)$.
        It implies $d(x,w) \le d(x,v)$.
    \end{proof}

    Based on Proposition~\ref{prop:DH}, a local-search algorithm for computing a center can be derived.
    For that, for any vertex $v$ that is not a center, we also need an efficient procedure for finding some vertex $u$ of smaller radius.
    Recall that every graph $G$ isometrically embeds in a smallest Helly graph $\mathcal{H}(G)$, which we call its injective hull.

    \begin{mylemma}[see Theorem $5$ in~\cite{InjHullGraphs}]\label{lem:DH-hull}
        For any distance-hereditary graph $G$, its injective hull $\mathcal{H}(G)$ is also distance-hereditary, and can be computed in $\cO(n+m)$ time.       
    \end{mylemma}

    By using Lemma~\ref{lem:DH-hull}, we can apply Lemma~\ref{lem:helly-for-DH}, for Helly graphs, directly to distance-hereditary graphs.
    The following result is also useful in the design of our algorithm:
    \begin{mylemma}[see Lemma $2.1$ in~\cite{Gpunimodal-ecc}]\label{lem:cartesian-tree}
        Let $G=(V,E)$ be a graph (possibly with loops), and let $\kappa : V \mapsto \mathbb{R}_{\ge 0}$.
        There is an $O(n+m)$-time algorithm that computes, for every $u \in V$, $\mu_\kappa(u) := \max\{\kappa(v) : v \nsim u \}$.
    \end{mylemma}

     We need one more result for choosing the start vertex of our local-search algorithm, namely:
    \begin{mylemma}[see Theorem $3$ in~\cite{Pri}]\label{lem:lambda-DH}
        Every distance-hereditary graph $3$-embeds into a tree.
    \end{mylemma}

    \begin{proof}[Proof of Theorem~\ref{thm:DH}]
        By the combination of Lemma~\ref{lem:lambda-DH} with Theorem~\ref{thm:hyp-weak}, some vertex $c_1$ s.t. $d(c_1,C_\pi(G)) \in \cO(1)$ can be computed in $\cO(m)$ time.
        Then, starting from $v_0 = c_1$, we construct some sequence $v_0,v_1,\ldots,v_\ell$ for which $r_\pi$ is monotonically decreasing, and $v_\ell \in C_\pi(G)$.
        Assume this sequence to be constructed up to some $i < \ell$.
        The following procedure is applied to its last vertex $v_i$:

        \smallskip
        \noindent
        {\bf High-level description of the procedure.}
        Assume first $v_i$ has some neighbor $u$ such that $r_\pi(u) < r_\pi(v_i)$.
        In this situation, we choose a neighbor $v_{i+1}$ of $v_i$ such that $r_\pi(v_{i+1})$ is minimum.
        Now assume for the remainder of the procedure $v_i$ is a local minimum of $r_\pi$ in $G$.
        Let us define the sets $F^j(v_i) = \{ z \in \Supp(\pi) : \pi(z)(d(v_i,z)+j) \ge r_\pi(v_i) \}$, for $0 \le j \le 2$.
        We compute the subset $U$ of all the vertices $u \in B_2(v_i)$ such that for any $j \in \{0,1,2\}$, for any $z \in F^j(v_i)$, $d(u,z) < d(v_i,z)+j$.
        If $U = \emptyset$, then we assert $v_i \in C_\pi(G)$.
        Otherwise, we compute $U^* = \bigcap\{ B_2(u) \cap U : u \in U \}$.
        Then, {\em in the injective hull $\mathcal{H}(G)$}, we compute some vertex $w$ such that $U^* \subseteq N_{\mathcal{H}(G)}[w]$ (whose existence follows from the Helly property).
        %let $w \in N_{\mathcal{H}(G)} \cap N_{\mathcal{H}(G)}(U)$ be minimizing $r_\pi(w)$.
        We stress that $w$ may be not in $G$.
        Finally, we choose for $v_{i+1}$ any vertex of $G$ that minimizes $r_\pi$ within $N_{\mathcal{H}(G)}[w] \cap V$.
        Furthermore, we assert $v_{i+1} \in C_\pi(G)$.

        \smallskip
        \noindent
        {\bf Correctness.} 
        We only need to consider the case when $v_i$ is a local minimum of $r_\pi$ in $G$.
        By Proposition~\ref{prop:DH}, $C_\pi(G) \subseteq B_2(v_i)$.
        For a vertex $u \in B_2(v_i)$ to satisfy $r_\pi(u) < r_\pi(v_i)$, it is necessary and sufficient that for any $j \in \{0,1,2\}$, for any $z \in F^j(v_i)$, $d(u,z) < d(v_i,z)+j$.
        Hence, the subset of all the vertices of $B_2(v_i)$ with radius less than $r_\pi(v_i)$ is exactly $U$.
        So, if $U = \emptyset$, then $v_i \in C_\pi(G)$.
        However, even if $U \ne \emptyset$, a vertex $u \in U$ satisfies $r_\pi(u) < r_\pi(v_i)$ but it may be the case that it is not a center.
        
        For any $u \in U$, we now prove that $C_\pi(G) \subseteq B_2(u)$.
        For that, let $w_u \in N_G(v_i) \cap N_G(u)$ be arbitrary.
        Suppose by contradiction there exists a $c \in C_\pi(G)$ such that $d(u,c) \ge 3$.
        We apply Lemma~\ref{lem:DH-outergates} in $G$ for the path $[v_i,w_u,u]$.
        By doing so, we obtain the existence of a $c^* \in I_G(c,v_i) \cap I_G(c,w_u) \cap I_G(c,u)$ such that $c^*$ has at least one neighbor among $\{v_i,w_u,u\}$.
        As $d(v_i,c) = 2$, either $c = c^* \sim w_u$, or $c \ne c^* \sim v_i$. 
        Furthermore, as $d(u,c) \geq 3$, we obtain $c^* \nsim u$.
        Therefore, $w_u \in S_1(u,c^*,G) \subseteq S_1(u,c,G)$.
        However, by Proposition~\ref{prop:DH} applied in $G$ for $u,c$, we get $r_\pi(w_u) \le r_\pi(u)$.
        As $r_\pi(u) < r_\pi(v_i)$, the latter contradicts our assumption that $v_i$ is a local minimum of $r_\pi$ in $G$.
        Consequently, $C_\pi(G) \subseteq B_2(u)$.
        As $u$ was chosen arbitrarily, and $C_\pi(G) \subseteq U$, we obtain $C_\pi(G) \subseteq U^*$.

        By construction, the vertices of $U^*$ are pairwise at distance at most two.
        It implies that the balls $B_1(u)$, for $u \in U^*$, pairwise intersect.
        By the Helly property applied for $\mathcal{H}(G)$, there exists a vertex $w$ in the injective hull such that $U^* \subseteq N_{\mathcal{H}(G)}[w]$.
        In particular, any vertex of $G$ that minimizes $r_\pi$ within $N_{\mathcal{H}(G)}[w] \cap V$ is a center.

        \smallskip
        \noindent
        {\bf Implementation.}
        We now provide an $\cO(m)$-time implementation of our procedure.
        (We will bound the number of calls to the procedure by $\cO(1)$ at the end of the proof.)
        For that, we first prove two intermediary claims:

        \begin{myclaim}\label{claim:dh-best-response}
            For any vertex $v$ in a distance-hereditary graph with $\cO(m)$ edges, for any profile $\pi$, we can compute the values $r_\pi(u)$, for each $u \in N(v)$, in total $\cO(m)$ time.            
        \end{myclaim}
            For each $u \in N(v)$, $r_\pi(u) = \max\{r^1(u),r^2(u),r^3(u),r^4(u),r^5(u)\}$, where
        \begin{itemize}
            \item $r^1(u) = \max\{\pi(v)\} \cup \{ \pi(u') : u' \in N(u) \cap N(v) \}$;
            \item $r^2(u) = \max\{ 2\pi(u'') : u'' \in N(v) \setminus N[u] \}$;
            \item $r^3(u) = \max\{ \pi(x)(d(v,x)-1) : x \nsim v, u \in I(v,x) \}$;
            \item $r^4(u) = \max\{ \pi(x)d(v,x) : x \nsim v, d(u,x) = d(v,x) \}$;
            \item $r^5(u) = \max\{ \pi(x)(d(v,x)+1) : x \nsim v, v \in I(u,x) \}$.
        \end{itemize}
        We can compute the values $r^1$ in total $\cO(m)$ time by scanning the neighborhoods $N(u)$, for $u \in N(v)$.
        Furthermore, we can also compute the values $r^2$ in total $\cO(m)$ time, as follows: we add a loop at every $u \in N(v)$, then we apply Lemma~\ref{lem:cartesian-tree} on the subgraph induced by $N[v]$, with $\kappa = 2\pi$.

        Now, for every $x \in \Supp(\pi)$ such that $x \nsim v$, let $x^* \in S_2(v,x)$.
        Note that all the values $x^*$ can be computed in $\cO(m)$ time while performing a BFS with start vertex $v$.
        Furthermore, by Lemma~\ref{lem:DH-outergates} applied for $N[v]$, $x^* \in I(x,u)$ holds for any $u \in N[v]$.
        It implies that for any $u \in N(v)$, $r^3(u) = \max\{ \pi(x)(d(v,x)-1) : x \nsim v, x^* \sim u \}$. 
        Therefore, the values $r^3$ can be also computed in total $\cO(m)$ time by scanning the neighborhoods $N(u)$, for $u \in N(v)$.

        For each $x \in \Supp(\pi)$ such that $x \nsim v$, we then consider some arbitrary $x^{**} \in N(v) \cap N(x^*)$.
        By Lemma~\ref{lem:DH-outergates} applied for the path $[v,x^{**},x^*]$, we obtain that for any $u \in N(v)$, $d(u,v) = d(u,x)$ holds if and only if $u \sim x^{**}$ and $u \nsim x^*$.
        However, this above characterization does not immediately lead to an $\cO(m)$-time procedure for computing the values $r^4$.
        Roughly, this is because we may have $x^{**} = y^{**}$ for some distinct vertices $x,y$.
        Therefore, for any $w \in N(v)$, we consider $X_w = \{ x^* : x^{**} = w \}$. Note that the sets $X_w$, for $w \in N(v)$, are pairwise disjoint.
        If $X_w \ne \emptyset$, then we construct a graph $H_w$ as follows:
        \begin{itemize}
            \item Every vertex of $N(v) \cap N(w)$ is in $H_w$;
            \item Every vertex of $X_w$ is also in $H_w$;
            \item There is no other vertex in $H_w$;
            \item And the edge $yz$ exists in $H_w$ if and only if (up to permuting $y$ and $z$): $y \in N(v) \cap N(w)$, $z \in X_w$, and $y,z$ are adjacent in the graph.
        \end{itemize}
        Define $\kappa_w$ such that $\kappa_w(u) = 0$ for any $u \in N(v) \cap N(w)$, and $\kappa_w(x^*) = \pi(x)d(v,x)$ for any $x^* \in X_w$.
        By applying Lemma~\ref{lem:cartesian-tree} for $H_w,\kappa_w$, for each $u \in N(v) \cap N(w)$ we obtain $\mu_w(u) = \max\{ \pi(x)d(v,x) : x^* \in X_w, x^* \nsim u \}$.
        As all the sets $X_w$, $w \in N(v)$, are disjoint, the cumulative number of edges in the graphs $H_w$ is in $\cO(m)$.
        Therefore, all the calls to Lemma~\ref{lem:cartesian-tree} can be done in total $\cO(m)$ time.
        By doing so, for any $u \in N(v)$, $r^4(u) = \max\{ \mu_w(u) : w \sim u, X_w \ne \emptyset \}$.

        Finally, consider the following graph $G_v$: the vertex set is $N[v]$ and, for any $w \in N(v)$, 
        $$N_{G_v}(w) = \bigcup\{ N(z) \cap N(v) : z \in \{w\} \cup X_w \}.$$
        We further add a loop at every vertex $w$.
        Define $\kappa_v(w) = \max\{0\} \cup \{ \pi(x)(d(v,x)+1) : x^* \in X_w \}$ for any $w \in N(v)$.
        To compute all the values $r^5$, we apply Lemma~\ref{lem:cartesian-tree} to $G_v,\kappa_v$. $\diamond$

        \medskip
        \noindent
         Note that in particular, we can use Claim~\ref{claim:dh-best-response} to compute a vertex $u$ that minimizes $r_\pi(u)$ within $N[v]$.
        Furthermore, as by Lemma~\ref{lem:DH-hull} $\mathcal{H}(G)$ is distance-hereditary, we can apply Claim~\ref{claim:dh-best-response} for both $G$ and $\mathcal{H}(G)$.
        The following claim completes the proof for the $\cO(m)$-time implementation of our procedure:
        
        \begin{myclaim}\label{claim:dh-balls-intersect}
            For any vertex $v$, for any $j \in \{0,1,2\}$, let $$F^j(v) = \{ x \in \Supp(\pi) : \pi(x)(d(v,x)+j) \ge r_\pi(v) \}.$$
            The following subset can be computed in $\cO(m)$ time: $$U^j(v) = \{ u : d(u,v) = 2, \text{and} \ \forall x \in F^j(v), d(u,x) < d(v,x) + j \}.$$
        \end{myclaim}
            If $v \in F^j(v)$, then as $j \le 2$, we return $U^j(v) = \emptyset$.
            Therefore, assume $v \notin F^j(v)$.
            We partition $F^j(v)$ in the following subsets: $M_1 := F^j(v) \cap N(v)$, $M_2 = F^j(v) \cap \{ x : d(v,x) = 2 \}$, and $M_3 = F^j(v) \cap \{ x : d(v,x) \ge 3 \}$.
            Furthermore, for each $x \in M_3$, let $x^* \in S_3(v,x)$ be arbitrary, and let $M_3^* = \{ x^* : x \in M_3 \}$.
            By Lemma~\ref{lem:DH-outergates}, for any $x \in M_3$, for any $u \in B_2(v)$, $x^* \in I(x,u)$ holds.
            Consequently, $U^j(v) = W_1 \cap W_2 \cap W_3$, where
            \begin{itemize}
                \item $W_1 = \bigcap\{ B_j(x) : x \in M_1 \}$;
                \item $W_2 = \bigcap\{ B_{j+1}(x) : x \in M_2\}$;
                \item $W_3 = \bigcap \{ B_{j+2}(x) : x \in M_3^* \}$.
            \end{itemize}
            We are left to explain how to compute $W_1$ (the computations of $W_2,W_3$ are similar).
            For that, we compute $\mathcal{H}(G)$, which takes $\cO(m)$ time (Lemma~\ref{lem:DH-hull}).
            Then, we apply Lemma~\ref{lem:helly-for-DH} for $\mathcal{H}(G)$, $M_1$, and $k = j$. $\diamond$

        \medskip
        \noindent
        Let us put Claims~\ref{claim:dh-best-response} and~\ref{claim:dh-balls-intersect} together.
        We apply Claim~\ref{claim:dh-best-response} in $G$ for $v := v_i$.
        By doing so, either we compute some neighbor $v_{i+1}$ such that $r_{\pi}(v_{i+1}) < r_\pi(v_i)$ is minimum within $N_G[v_i]$, or we assert that $v_i$ is a local minimum of $r_\pi$ in $G$.
        In the latter case, we apply Claim~\ref{claim:dh-balls-intersect} in $G$ for $v := v_i$ and for each $j \in \{0,1,2\}$.
        By doing so, the set $U$ is computed in $\cO(m)$ time.
        Recall that $\mathcal{H}(G)$ can be computed in $\cO(m)$ time (Lemma~\ref{lem:DH-hull}).
        If $U \ne \emptyset$, then we apply Lemma~\ref{lem:helly-for-DH} for $\mathcal{H}(G)$, $M = U$ and $k =2$ to compute the subset $U^*$.
        To compute a vertex $w$ of $\mathcal{H}(G)$ such that $U^* \subseteq N_{\mathcal{H}(G)}[w]$, it suffices to scan the closed neighborhoods $N_{\mathcal{H}(G)}[u]$, for $u \in U^*$.
        Finally, we are done applying Claim~\ref{claim:dh-best-response} in $\mathcal{H}(G)$ for $v := w$.

        \smallskip
        \noindent
        {\bf Runtime analysis.} It remains to prove that the number of calls to the procedure is in $\cO(1)$.
        Observe that if we call the procedure for some $v_i$, then the procedure is called at least one more time if and only if there exists some $u \in N(v_i)$ such that $r_\pi(u) < r_\pi(v_i)$.
        Therefore, to bound the number of calls to the procedure, it suffices to bound the maximum length of a path $[v_0,v_1,\ldots,v_s]$ for which $r_\pi$ is monotonically decreasing and such that, for any $i < s$, $v_{i+1}$ minimizes $r_\pi$ within $N(v_i)$.

        For that, let $c \in C_\pi(G)$ be closest to $v_0$.

        \begin{myclaim}\label{claim:dh-distgeq2}
            Consider any $i$ such that $c$ is also closest to $v_i$ within $C_\pi(G)$.
            If $d(v_i,c) \ge 2$, then, $d(v_{i+1},c) \le d(v_i,c)$ holds.
        \end{myclaim}
        Suppose by contradiction $d(v_i,c) < d(v_{i+1},c)$ holds.
        It implies $d(v_{i+1},c) \ge 3$.
        However, as $v_i \in S_1(v_{i+1},c)$, by Proposition~\ref{prop:DH}, we would obtain $r_\pi(v_i) \le r_\pi(v_{i+1})$.
        A contradiction. $\diamond$

        \begin{myclaim}\label{claim:dh-far}
            Consider any $i$ such that $c$ is also closest to $v_i$ within $C_\pi(G)$.
            If $d(v_i,c) \ge 3$, then $d(v_{i+1},c) < d(v_i,c)$.
            In particular, $c$ is also closest to $v_{i+1}$ within $C_\pi(G)$.
        \end{myclaim}
            Suppose $d(v_i,c) \le d(v_{i+1},c)$.
            By Claim~\ref{claim:dh-distgeq2}, $d(v_i,c) = d(v_{i+1},c)$ holds.
            By Lemma~\ref{lem:DH-outergates} applied for the edge $v_iv_{i+1}$, there exists a $w \in S_1(v_i,c) \cap S_1(v_{i+1},c)$.
            But as $d(v_{i+1},C_\pi(G)) \ge d(v_i,C_\pi(G))-1 = d(v_i,c)-1 \ge 2$, we get that $v_{i+1}$ is not a center.
            Hence, by Proposition~\ref{prop:DH}, $r_\pi(w) < r_\pi(v_{i+1})$ holds, thus contradicting the minimality of $r_\pi(v_{i+1})$ within $N(v_i)$. $\diamond$

        \begin{myclaim}\label{claim:dh-close}
            Let $v_i,v_{i+1},\ldots,v_{i+p}$ be satisfying $d(v_{i+j},c) = 2$ for any $0 \le j \le p$.
            Then, $p \le 4$.
        \end{myclaim}
            Indeed, as $r_\pi$ is monotonically decreasing on the path $v_i,v_{i+1},\ldots,v_{i+p}$, the latter must be an induced path.       
            As $G$ is distance-hereditary, this path is shortest. It implies $d(v_i,v_{i+p}) = p$.
            Furthermore, $d(v_i,v_{i+p}) \le d(v_i,c) + d(c,v_{i+p}) = 4$. $\diamond$

        \begin{myclaim}\label{claim:dh-last}
           If $d(v_i,c) = 1$, then either $i=s$, or $i = s-1$ and $v_s \in C_\pi(G)$. 
        \end{myclaim}
            Assume $i < s$.
            As $c \in C_\pi(G)$, and $v_{i+1}$ minimizes $r_\pi$ within $N(v_i)$, we obtain $v_{i+1} \in C_\pi(G)$.
            Then, necessarily $i+1=s$. $\diamond$

    \smallskip
    \noindent
        Let us put the claims altogether to prove that the path has length $s \le d(v_0,c)+4 = d(v_0,C_\pi(G))+4$.
        As $d(v_0,C_\pi(G)) \in \cO(1)$, the latter will prove that $s = \cO(1)$.
        If $d(v_0,c) \le 1$, then, by Claim~\ref{claim:dh-last}, $s \le 1$ holds.
        Assume now $d(v_0,c) = 2$.
        We consider the largest index $j$ such that $d(v_{j'},c) = 2$ holds for each $j'$ between $0$ and $j$.
        By Claim~\ref{claim:dh-close}, $j \le 4$.
        Assume further $j < s$ (for else, we are done).
        If $c$ is not a closest center to $v_j$, then, $d(v_j,c') \le 1$ holds for some other $c' \in C_\pi(G)$, and so, by Claim~\ref{claim:dh-last}, $s \le j+1 \le 5$.         
        Otherwise, by Claim~\ref{claim:dh-distgeq2}, $d(v_{j+1},c) \le d(v_j,c) = 2$ holds.
        In the latter case, the maximality of index $j$ implies $d(v_{j+1},c) = 1$, and so, by Claim~\ref{claim:dh-last}, $s \le j+2 \leq 6 = d(v_0,c)+4$.
        Finally, assume $d(v_0,c) \ge 3$.
        We consider the largest index $k$ such that $d(v_{k'},c) \ge 3$ for each $k'$ between $0$ and $k$.
        By Claim~\ref{claim:dh-far}, $d(\cdot,c)$ is monotonically decreasing on the sub-path $[v_0,v_1,\ldots,v_k]$; moreover, for any $k'$ such that $0 \le k' \le k$, $c$ is also a closest center to $v_{k'}$.
        It implies $k = d(v_0,c) - d(v_k,c) \le d(v_0,c) - 3$.
        As before, we further assume $k < s$ (for else, we are done).
        By Claim~\ref{claim:dh-far}, $d(v_{k+1},c) < d(v_k,c)$ holds.
        The maximality of $k$ implies that $d(v_{k+1},c) = d(v_k,c)-1 = 2$.
        If $c$ is not a closest center to $v_{k+1}$, then, $d(v_{k+1},c'') \le 1$ holds for some other $c'' \in C_\pi(G)$, and so, by Claim~\ref{claim:dh-last}, $s \le k+2 \le d(v_0,c)-1$. 
        Otherwise, by using the same arguments as above (where we replace the start vertex $v_0$ with $v_{k+1}$), we get $s-(k+1) \le 6$.
        Hence, $s = k+1+(s-k-1) \le (d(v_0,c)-3)+1+6 = d(v_0,c)+4$.
    \end{proof}

    \subsection{Planar graphs.}\label{planar}

\begin{mytheorem}\label{th:planar}
    There is an $2^{\cO(\delta)}n\log^{10}{n}$-time algorithm for the \textsc{Weighted Center} problem on $\delta$-hyperbolic planar graphs.
\end{mytheorem}
    
Recall that a \emph{separator} of a graph $G$ is a subset $S$ such that $G \setminus S$ is disconnected.
The \emph{balance} of a separator $S$ is the largest $\alpha$ such that every connected component of $G \setminus S$ contains at most $(1-\alpha)n$ vertices.
Finally, a separator $S$ is called \emph{isometric} if it induces an isometric (or distance-preserving) subgraph of $G$.
Note that if $S$ is an isometric separator, and $C$ is any connected component of $G \setminus S$, then $C \cup S$ also induces an isometric subgraph.
\begin{mylemma}[see Theorem $1$ in~\cite{PlanarHyp}]\label{lem:planar-sep}
    If $G$ is a planar $\delta$-hyperbolic graph, then we can compute in $\cO(\delta^2n\log^4{n})$ time an isometric separator $S$ such that
    \begin{enumerate}
        \item either $S$ is a path of length $\cO(\delta^2\log{n})$, and it has constant balance;
        \item or $S$ is a cycle of length $\cO(\delta)$, and it has balance at least $\frac 1 {2^{\cO(\delta)}\log{n}}$.
    \end{enumerate}
\end{mylemma}

As it can be expected, our strategy for proving Theorem~\ref{th:planar} consists in recursively disconnecting a planar hyperbolic graph using Lemma~\ref{lem:planar-sep}.
After each separation, for any connected component $C$, we need to compute, for each vertex of $C$, its maximum weighted distance to the vertices in the other connected components.
For that, we use the algorithmic toolkit from~\cite{VCdim}, which is based on VC-theory.
In what follows, for any graph $G$, the \emph{boundary} of an induced subgraph $H$ is defined as the subset of all the vertices of $H$ with at least one neighbor in $V(G) \setminus V(H)$.
We denote the boundary of $H$ by $\partial H$. Furthermore, we assume for what follows that $\partial H$ is totally ordered: $\partial H = \langle s_0,s_1,\ldots,s_{|\partial H|-1} \rangle$ (the ordering is arbitrary).
Then, for any vertex $v$, the \emph{pattern} of $v$ with respect to $\partial H$ is the $|\partial H|$-dimensional array $\texttt{p}_v$ such that, for each $i$, $\texttt{p}_v[i] = d(v,s_i) - d(v,s_0)$.
The following result is proved in~\cite{VCdim} for $K_h$-minor-free graphs, but we restate it only for planar graphs:
\begin{mylemma}[see Lemma $4$ in~\cite{VCdim}]\label{lem:pattern-count}
    Let $H$ be a connected induced graph of some planar graph $G$.
    The number of distinct patterns with respect to $\partial H$ is in $\cO(|\partial H||V(H)|^4)$.
\end{mylemma}

Since we are dealing with arbitrary profiles (not necessarily binary), in what follows we will need to combine Lemma~\ref{lem:pattern-count} with the following convex hull trick:
\begin{mylemma}[see Lemma $4.2$ in~\cite{DucCW}]\label{lem:convex-hull-trick}
    Let $F$ be a set of $n$ linear functions $f_i : t \mapsto a_i \cdot t + b_i$, where $a_i,b_i \ge 0$. 
    Then after an $\cO(n\log{n})$-time pre-processing, for any $x \ge 0$ we can compute $\max\{a_i \cdot x +b_i : 1 \le i \le n \}$ in $\cO(\log{n})$ time.
\end{mylemma}

%We are now ready to prove the main result of this part:
\begin{proof}[Proof of Theorem~\ref{th:planar}]
    Assume for what follows $n \ge 2^{\Theta(\delta)}$ (otherwise, we solve the \textsc{Weighted Center} problem by brute-force in $\cO(n^2) = 2^{\cO(\delta)}$ time).
    We apply Lemma~\ref{lem:planar-sep} to compute an isometric separator $S$ in $\cO(\delta^2n\log^4{n})$ time.
    Furthermore, as $|S| = \cO(\delta^2\log{n})$, we can afford to compute $r_\pi(s)$, for each $s \in S$; it takes $\cO(\delta^2n\log{n})$ time in total.
    For the other vertices, let us consider the connected components $C_1,C_2,\ldots,C_q$, $q \ge 2$, of $G \setminus S$.
    For each $i$ such that $1 \le i \le q$, let $\pi_i$ denote the restriction of $\pi$ to $C_i \cup S$.
    Then, for any vertex $v$ of $C_i$, $r_\pi(v) = \max\{ r_{\pi_i}(v) \}\cup\{ \pi(u)d(v,u) : u \notin C_i \cup S \}$.
    As $C_i \cup S$ induces an isometric subgraph $G_i$ of $G$, the graph $G_i$ is also planar and $\delta$-hyperbolic.
    In particular, all the values $r_{\pi_i}(v)$, for $v \in C_i$, can be computed by applying our algorithm recursively on $G_i$.

    For each $i$ such that $1 \le i \le q$, for each $v \in C_i$, define $\lambda(v) = \max\{\pi(u)d(v,u) : u \notin C_i \cup S\}$.
    By the above, $r_\pi(v) = \max\{r_{\pi_i}(v),\lambda(v)\}$.
    So, we are left to compute the values $\lambda(v)$, for each $v \in V \setminus S$. 
    For that, for each $j$ with $1 \le j \le \lceil\log{q}\rceil$, let $X_j,Y_j$ be such that:
    \begin{itemize}
        \item $X_j$ contains all components $C_i$ such that the $j^{th}$ bit of $i$ is equal to $0$;
        \item $Y_j$ contains all components $C_i$ such that the $j^{th}$ bit of $i$ is equal to $1$.
    \end{itemize}
    Note that for each pair of vertices $u,v$ that are in different components of $G \setminus S$, there is at least one $j$ such that either $u \in X_j, v \in Y_j$ or $u \in Y_j, v \in X_j$.
    Therefore, to compute all values $\lambda(v)$, it suffices to compute the following values $\lambda_j(v)$ for any $j$: if $v \in X_j$, then $\lambda_j(v) = \max\{ \pi(u)d(v,u) : u \in Y_j \}$; and if $v \in Y_j$, then $\lambda_j(v) = \max\{ \pi(u)d(v,u) : u \in X_j \}$.
    By symmetry, it suffices to present an algorithm for computing all values $\lambda_j(v)$, for $v \in X_j$.
      
    For that, let $\partial S = s_0,s_1,\ldots,s_{|\partial S|-1}$.
    We compute all profiles $\texttt{p}_u$, for $u \in Y_j$, with respect to $\partial S$.
    This can be done in $\cO(|\partial S|n) = \cO(\delta^2n\log{n})$ by performing a BFS for each $s_i \in \partial S$.
    Denote by $P_S$ the set of all distinct profiles with respect to $\partial S$. Note that $P_S$ can be computed in $\cO(|\partial S|n\log{n}) = \cO(\delta^2n\log^2{n})$ time, by sorting.
    Furthermore, by Lemma~\ref{lem:pattern-count}, it holds $|P_S| = \cO(|\partial S||S|^4) = \cO(|S|^5) = \cO(\delta^{10}\log^5{n})$.
    For each $\texttt{p} \in P_S$, let $Y_{j,\texttt{p}}$ contain all the vertices $u \in Y_j$ such that $\texttt{p}_u = \texttt{p}$.
    For any $v \in X_j$, let $\lambda_{j,\texttt{p}}(v) = \max\{\pi(u)d(v,u) : u \in Y_{j,\texttt{p}}\}$.
    To compute $\lambda_j(v)$, it suffices to compute $\lambda_{j,\texttt{p}}$ for each $\texttt{p} \in P_S$.

    So, let $\texttt{p} \in P_S$ be fixed.
    For each $k$ with $0 \le k \le |\partial S|-1$, we define the family of functions $F_{j,\texttt{p},k} = \{ t \mapsto \pi(u) \cdot t + \pi(u)d(s_k,u) : u \in Y_{j,\texttt{p}} \}$, for which we apply Lemma~\ref{lem:convex-hull-trick}. It takes $\cO(|\partial S|n\log{n}) = \cO(\delta^2n\log^2{n})$ time in total.
    Then, we consider each $v \in X_j$ sequentially.
    Let $k$ be such that $d(v,s_k)+\texttt{p}[k]$ is minimum.
    We claim that for any $u \in Y_{j,\texttt{p}}$, $s_k \in I(v,u)$.
    To prove this claim, as any shortest $(v,u)$-path contains a vertex of $\partial S$, it suffices to prove that $d(v,s_k)+d(s_k,u)$ is minimum.
    That is indeed the case because for any $k'$ with $0 \le k' \le |\partial S|-1$,
    %\begin{flalign*}
        $d(v,s_k)+d(s_k,u) - d(s_0,u) = d(v,s_k) + \texttt{p}[k] \le d(v,s_{k'}) + \texttt{p}[k'] =  d(v,s_{k'})+d(s_{k'},u) - d(s_0,u)$.
    %\end{flalign*}
    %The claim follows from adding $d(s_0,u)$ on both sides of the inequality.
    Consequently, by Lemma~\ref{lem:convex-hull-trick} applied for $F_{j,\texttt{p},k}$ and $x = d(v,s_k)$, we can compute $\lambda_{j,\texttt{p}}(v)$ in $\cO(\log{n})$ time.

    Overall, the values $\lambda_{j,\texttt{p}}(v)$, for $v \in X_j$, can be computed in $\cO(\delta^2n\log^2{n})$ time.
    It implies that the values $\lambda_j(v)$, for $v \in X_j$, can be computed in $\cO(|P_S|\delta^2n\log^2{n}) = \cO(\delta^{12}n\log^7{n})$ time.
    Hence, the values $\lambda(v)$, for $v \in V \setminus S$, can be computed in $\cO(\delta^{12}n\log^8{n})$ time.

    \smallskip
    Finally, by Lemma~\ref{lem:planar-sep}, the balance of $S$ is at least $1/(2^{\cO(\delta)}\log{n})$.
    Furthermore, as we assume $n \ge 2^{\Theta(\delta)}$, the inclusion of $S$ in all the $G_i$'s has a marginal impact on the size of these subgraphs.
    Therefore, the recursion depth of our algorithm is at most $2^{\cO(\delta)}\log^2{n}$.
\end{proof}

\section{Perspectives}\label{perspectives}

The main contribution of the paper is an algorithm for computing on any $\delta$-hyperbolic graph with $m$ edges, for any profile $\pi$, some vertex at distance $\cO(\delta)$ to the center in $\Tilde{\cO}(\delta m)$ time (Theorem~\ref{thm:hyp-strong}).
In doing so, we answer one of the questions left open in~\cite{Gpunimodal-ecc} (Question $10.1$).
Our dependency on $\delta$ is linear, but it is quite large.
Improving this constant could be an interesting challenge.

Based on this main result, a framework for solving the \textsc{Weighted Center} problem exactly on some classes of hyperbolic graphs is proposed. 
It implies optimal, or almost optimal, algorithms for the \textsc{Weighted Center} problem on chordal graphs, chordal bipartite graphs, dually chordal graphs and distance-hereditary graphs.
%The latter results provide a partial answer to yet another question left open in~\cite{Gpunimodal-ecc} (Question $10.2$).
However, the inclusion to our framework of AT-free graphs, graphs of bounded asteroidal number, and other classic examples of hyperbolic graphs, remains to be done.

Finally, we present parameterized almost linear-time algorithms for the \textsc{Weighted Center} problem on planar graphs and several other classes of graphs, where the parameter is the hyperbolicity.
We leave for future work whether similar hyperbolicity-based parameterizations can be also achieved for some other distance-related problems on these classes of graphs.

\subsection*{Acknowledgements} This work has been supported in part by the PHC Br\^{a}ncu\c{s}i LoSST.

\bibliographystyle{amsplain}
\bibliography{biblio.bib}

@inproceedings{AVW16,
    author = {A.~Abboud and V.~Vassilevska Williams and J.~Wang},
    title = {Approximation and fixed parameter subquadratic algorithms for radius and diameter in sparse graphs},
    booktitle = {Proceedings of the ACM-SIAM Symposium on Discrete Algorithms (SODA)},
    pages = {377–-391},
    year = {2016}
}

@article{MedianOfMedians,
    author = {L. B\'en\'eteau and J. Chalopin and V. Chepoi and Y. Vax\`es},
    title = {Medians in median graphs and their cube complexes in linear time},
    journal = {Journal of Computer and System Sciences},
    volume = {126},
    pages = {80--105},
    year = {2022}
}

@article{WP,
    author = {L. B\'en\'eteau and J. Chalopin and V. Chepoi and Y. Vax\`es},
    title = {Graphs with {$G^p$}-connected medians},
    journal = {Mathematical Programming, Ser B},
    volume = {203},
    pages = {369--420},
    year = {2024}
}

@inproceedings{BDHMedian,
    author = {P. Berg\'e and G. Ducoffe and M. Habib},
    title = {Quasilinear-time eccentricities computation, and more, on median graphs},
    booktitle = {Proceedings of the ACM-SIAM Symposium on Discrete Algorithms (SODA)},
    pages = {1679--1704},
    year = {2025}
}

@inproceedings{MedianOfMediansArray,
    author = {M. Blum and R.W. Floyd and V. Pratt and R.L. Rivest and R.E. Tarjan},
    title = {Linear time bounds for median computations},
    booktitle = {Proceedings of the ACM Symposium on Theory of Computing (STOC)},
    pages = {119-124},
    year = {1972}
}

@article{BrChDrDuallyChordal,
    author = {A. Brandst\"{a}dt and V. Chepoi and F.F. Dragan},
    title = {The algorithmic use of hypertree structure and maximum neighbourhood orderings},
    journal = {Discrete Applied Mathematics},
    volume = {82},
    pages = {43--77},
    year = {1998}
}

@article{BrChDr,
    author = {A. Brandst\"{a}dt and V. Chepoi and F.F. Dragan},
    title = {{Distance Approximating Trees for Chordal and Dually Chordal Graphs}},
    journal = {Journal of Algorithms},
    volume = {30},
    pages = {166-184},
    year = {1999}
}

@article{HypChordal,
    author = {G. Brinkmann and J.H. Koolen and V. Moulton},
    title = {On the hyperbolicity of chordal graphs},
    journal = {Annals of Combinatorics},
    volume = {5},
    number = {1},
    pages = {61--69},
    year = {2001}
}

@BOOK{BoMu08,
	title={Graph theory},
	author={J.A. Bondy and U.S.R. Murty},
    publisher={Graduate Texts in Mathematics},
	year={2008}
}

@article{Cab,
    author = {S. Cabello},
    title = {Subquadratic algorithms for the diameter and the sum of pairwise distances in planar graphs},
    journal = {ACM Transactions on Algorithms},
    volume = {15},
    number = {2},
    articleno = {21},
    year = {2018}
}

@TECHREPORT{Gpunimodal-ecc,
    author = {J. Chalopin and V. Chepoi and F.F. Dragan and G. Ducoffe and Y. Vaxx\`es},
    title = {On {$G^p$}-unimodality of radius functions in graphs: structure and algorithms},
    institution = {arXiv},
    number = {2503.15011},
    year = {2025}
}

@inproceedings{ChanVCdim,
    author = {T.M. Chan and H.-C. Chang and J. Gao and S. Kisfaludi-Bak and H. Le and D.W. Zheng},
    title = {{Truly Subquadratic Time Algorithms for Diameter and Related Problems in Graphs of Bounded VC-dimension}},
    booktitle = {Proceedings of the IEEE Symposium on Foundations of Computer Science (FOCS)},
    pages = {2728 -- 2765},
    year = {2025}
}

@inproceedings{ChDrChordal,
    author = {V. Chepoi and F.F. Dragan},
    title = {A linear algorithm for computing a central vertex of a chordal graph},
    booktitle = {Proceedings of the European Symposium on Algorithms (ESA)},
    pages = {159--170},
    year = {1994}
}

@article{ChDr,
    author = {V. Chepoi and F.F. Dragan},
    title = {{A Note on Distance Approximating Trees in Graphs}},
    journal = {European Journal of Combinatorics},
    volume = {21},
    pages = {761-766},
    year = {2000}
}

@inproceedings{ChDrEsHaVa,
    author = {V. Chepoi and F.F. Dragan and B. Estellon and M. Habib and Y. Vax\`es},
    title = {Diameters, centers, and approximating trees of $\delta$-hyperbolic geodesic spaces and graphs},
    booktitle = {Proceedings of the Symposium on Computational Geometry (SoCG)},
    pages = {59–-68},
    year = {2008}
}

@article{ChDrNeRaVa,
    author = {V. Chepoi and F.F. Dragan and I. Newman and Y. Rabinovich and Y. Vax\`es},
    title = {Constant approximation algorithms for embedding graph metrics into trees and outerplanar graphs},
    journal = {Discretre \& Computational Geometry},
    volume = {47},
    pages = {187--214},
    year = {2012}
}

@inproceedings{ChEs,
    author = {V. Chepoi and B. Estellon},
    title = {Packing and covering $\delta$-hyperbolic spaces by balls},
    booktitle = {Proceedings of the International Workshops on Approximation Algorithms for Combinatorial Optimization Problems, and on Randomization and Computation (APPROX-RANDOM)},
    pages = {59--73},
    year = {2007}
}

@article{ChOs,
    author = {V. Chepoi and D. Osajda},
    title = {Dismantlability of weakly systolic complexes and applications},
    journal = {Transactions of the American Mathematical Society},
    volume = {367},
    number = {2},
    pages = {1247--1272},
    year = {2015}
}

@book{Cormen,
    author = {T.H. Cormen and C.E. Leiserson and R.L. Rivest and C. Stein},
    title = {Introduction to algorithms},
    publisher = {MIT press},
    year = {2022}
}

@article{PFPT,
    author = {D. Coudert and G. Ducoffe and A. Popa},
    title = {{Fully Polynomial FPT Algorithms for Some Classes of Bounded Clique-width Graphs}},
    journal = {ACM Transactions on Algorithms},
    volume = {15},
    number = {3},
    articleno = {33},
    year = {2019}
}

@article{HypAlg1,
    author = {B. Das Gupta and M. Karpinski and N. Mobasheri and F. Yahyanejad},
    title = {{Effect of Gromov-hyperbolicity parameter on cuts and expansions in graphs and some algorithmic implications}},
    journal = {Algorithmica},
    volume = {80},
    number = {2},
    pages = {772--800},
    year = {2018}
}

@article{BucketQueue,
    author = {R.B. Dial},
    title = {Algorithm 360: shortest-path forest with topological ordering [H]},
    journal = {Communications of the ACM},
    volume = {12},
    number = {11},
    pages = {632 -- 633},
    year = {1969}
}

@inproceedings{DrDH,
    author = {F.F. Dragan},
    title = {Dominating cliques in distance-hereditary graphs},
    booktitle = {Proceedings of the Scandinavian Workshop on Algorithm Theory (SWAT)},
    pages = {370--381},
    year = {1994}
}

@article{DrGu_hyp,
    author = {F.F. Dragan and H.M. Guarnera},
    title = {Eccentricity terrain of $\delta$-hyperbolic graphs},
    journal = {Journal of Computer and System Sciences},
    volume = {112},
    pages = {50--65},
    year = {2020}
}

@article{HellyHyp,
    author = {F.F. Dragan and G. Ducoffe and H.M. Guarnera},
    title = {{Fast deterministic algorithms for computing all eccentricities in (hyperbolic) Helly graphs}},
    journal = {Journal of Computer and System Sciences},
    volume = {199},
    articleno = {103606},
    year = {2025}
}

@article{dragan2026certificates,
  title={{Certificates in P and subquadratic-time computation of radius, diameter, and all eccentricities in graphs}},
  author={Dragan, Feodor and Ducoffe, Guillaume and Habib, Michel and Viennot, Laurent},
  journal={Algorithmica},
  volume={88},
  number={1},
  pages={13},
  year={2026},
  publisher={Springer}
}

@article{Dress,
    author = {A.W.M. Dress},
    title = {Trees, tight extensions of metric spaces, and the cohomological dimension of certain groups: a note on combinatorial properties of metric spaces},
    journal = {Advances in Mathematics},
    volume = {53},
    number = {3},
    pages = {321--402},
    year = {1984}
}

@inproceedings{AbsoluteRetract,
    author = {G. Ducoffe},
    title = {{Beyond Helly Graphs: The Diameter Problem on Absolute Retracts}},
    booktitle = {Proceedings of the International Workshop on Graph-Theoretic Concepts in Computer Science (WG)},
    pages = {321--335},
    year = {2021}
}

@article{DucCW,
    author = {G. Ducoffe},
    title = {{Optimal Centrality Computations Within Bounded Clique-Width Graphs}},
    journal = {Algorithmica},
    volume = {84},
    pages = {3192--3222},
    year = {2022}
}

@article{kHelly,
    author = {G. Ducoffe},
    title = {{Distance problems within Helly graphs and $k$-Helly graphs}},
    journal = {Theoretical Computer Science},
    volume = {946},
    articleno = {113690},
    year = {2023}
}

@article{Helly,
    author = {G. Ducoffe and F.F. Dragan},
    title = {{A story of diameter, radius and (almost) Helly property}},
    journal = {Networks},
    volume = {77},
    number = {3},
    pages = {435--453},
    year = {2021}
}

@article{Dyer,
    author = {M.E. Dyer},
    title = {{On A Multidimensional Search Technique And Its Application To The Euclidean One-Centre Problem}},
    journal = {SIAM Journal on Computing},
    volume = {15},
    number = {3},
    pages = {725--738},
    year = {1986}
}

@article{GolRotCW,
    author = {M.C. Golumbic and U. Rotics},
    title = {{On The Clique-Width Of Some Perfect Graph Classes}},
    journal = {International Journal of Foundations of Computer Science},
    volume = {11},
    number = {3},
    pages = {423--443},
    year = {2000}
}

@INBOOK{Gr,
	Author = {M. Gromov},
	Booktitle = {Essays in Group Theory},
	Title = {{Hyperbolic Groups}},
	Year = {1987}
 }

@article{InjHullGraphs,
    author = {H.M. Guarnera and F.F. Dragan and A. Leitert},
    title = {Injective hulls of various graph classes},
    journal = {Graphs and Combinatorics},
    volume = {38},
    number = {4},
    articleno = {112},
    year = {2022}
}

@article{KaHa,
    author = {O. Kariv and S.L. Hakimi},
    title = {{An Algorithmic Approach To Network Location Problems. I: The $p$-Centers}},
    journal = {SIAM Journal on Applied Mathematics},
    volume = {37},
    number = {3},
    pages = {513--538},
    year = {1979}
}

@inproceedings{KenSanNa,
    author = {W.S. Kennedy and I. Saniee and O. Narayan},
    title = {On the hyperbolicity of large-scale networks and its estimation},
    booktitle = {Proceedings of the IEEE International Conference on Big Data (BigData)},
    pages = {3344--3351},
    year = {2016}
}

@inproceedings{PlanarHyp,
    author = {S. Kisfaludi-Bak and J. Masa\u{r}\'ikova and E.J. van Leeuwen and B. Walczak and K. W\c{e}grzycki},
    title = {{Separator Theorem and Algorithms for Planar Hyperbolic Graphs}},
    booktitle = {Proceedings of the Symposium on Computational Geometry (SoCG)},
    articleno = {67},
    year = {2024}
}

@article{HypBridged,
    author = {J.H. Koolen and V. Moulton},
    title = {Hyperbolic bridged graphs},
    journal = {European Journal of Combinatorics},
    volume = {23},
    number = {6},
    pages = {683--699},
    year = {2002}
}

@inproceedings{HypAlg2,
    author = {R. Krauthgamer and J.R. Lee},
    title = {Algorithms on negatively curved spaces},
    booktitle = {Proceedings of the IEEE Symposium on Foundations of Computer Science (FOCS)},
    pages = {119–132},
    year = {2006}
}

@inproceedings{VCdim,
    author = {H. Le and C. Wulff-Nilsen},
    title = {{VC set systems in minor-free (di) graphs and applications}},
    booktitle = {Proceedings of the ACM-SIAM Symposium on Discrete Algorithms (SODA)},
    pages = {5332 -- 5360},
    year = {2024}
}

@article{Me2,
    author = {N. Megiddo},
    title = {{Linear time algorithms for linear programming in ${\mathbb R}^3$ and related problems}},
    journal = {SIAM Journal on Computing},
    volume = {12},
    pages = {759-–776},
    year = {1983}
}

@article{Pap,
    author = {P. Papasoglu},
    title = {Polynomial growth and asymptotic dimension},
    journal = {Israel Journal of Mathematics},
    volume = {255},
    pages = {985--1000},
    year = {2023}
}

@inproceedings{Pri,
    author = {E. Prisner},
    title = {Distance approximating spanning trees},
    booktitle = {Proceedings of the Symposium on Theoretical Aspects of Computer Science (STACS)},
    pages = {499--510},
    year = {1997}
}

@article{Sh,
    author = {V. Shchur},
    title = {{A quantitative version of the Morse lemma and quasi-isometries fixing the ideal boundary}},
    journal = {Journal of Functional Analysis},
    volume = {264},
    number = {3},
    pages = {815--836},
    year = {2013}
}

@article{SoCh1983,
    author = {V.P. Soltan and V.D. Chepoi},
    title = {Conditions for invariance of set diameters under $d$-convexification in a graph},
    journal = {Cybernetics},
    volume = {19},
    pages = {750--756},
    year = {1983},
    note = {Russian, English transl.}
}

@article{WuZh01,
    author = {Y. Wu and C. Zhang},
    title = {Hyperbolicity and chordality of a graph},
    journal = {The Electronic Journal of Combinatorics},
    volume = {18},
    number = {1},
    pages = {P43},
    year = {2001}
}

\end{document}